\documentclass{revtex4}
\usepackage{amsmath,amsthm}
\usepackage{rotating}
\usepackage{multirow}
\usepackage{graphicx}

\usepackage{amsfonts}
\newtheorem{theorem}{Theorem}[section]
\newtheorem{lemma}[theorem]{Lemma}
\newtheorem{proposition}[theorem]{Proposition}
\newtheorem{cor}[theorem]{Corollary}

\theoremstyle{remark}
\newtheorem{remark}[theorem]{Remark}

\theoremstyle{definition}
\newtheorem{definition}[theorem]{Definition}

\theoremstyle{example}
\newtheorem{example}[theorem]{Example}

\theoremstyle{notation}

\newcommand{\bra}[1]{\langle#1|}
\newcommand{\ket}[1]{|#1\rangle}

\begin{document}

\title{Comonotonicity and Choquet integrals of Hermitian operators and their applications}            
\author{A. Vourdas}
\affiliation{Department of Computing,\\
University of Bradford, \\
Bradford BD7 1DP, United Kingdom\\a.vourdas@bradford.ac.uk}

\begin{abstract}
In a quantum system with $d$-dimensional Hilbert space,
the $Q$-function of a Hermitian positive semidefinite operator $\theta$,  
is defined in terms of the $d^2$ coherent states in this system.
The Choquet integral ${\cal C}_Q(\theta)$ of the  $Q$-function of $\theta$, is introduced using a ranking of the values of the $Q$-function, and 
M\"obius transforms which remove the overlaps between coherent states.
It is a figure of merit of the quantum properties of Hermitian operators, 
and it provides upper and lower bounds to various physical quantities in terms of the $Q$-function.
Comonotonicity is an important concept in the formalism, which is used to formalize the vague concept of
physically similar operators. 
Comonotonic operators are shown to be bounded,
with respect to an order based on Choquet integrals.
Applications of the formalism to the study of the ground state of a physical system, are discussed.
Bounds for partition functions, are also derived. 
\end{abstract}
\maketitle

\section{Introduction}
There are many quantities which describe quantum properties of quantum systems.
The various entropic quantities (von Neumann entropy, Wehrl entropy \cite{We}, etc) are examples of this.
In this paper we introduce Choquet integrals as indicators of the quantum properties of Hermitian operators.
Choquet integrals are used in problems with probabilities, where the various alternatives are not independent, but they overlap with each other.

We consider a quantum system $\Sigma (d)$ with variables in ${\mathbb Z}(d)$ (the integers modulo $d$), described with the 
$d$-dimensional Hilbert space $H(d)$ \cite{Fin,Fin2}.
Let $\Omega$ be the set of $d^2$ coherent states, associated with the Heisenberg-Weyl group of displacements in this discrete system.
The overlapping nature of coherent states is our motivation for the use of Choquet integrals.
We map the $Q$ function of a Hermitian positive semidefinite operator $\theta$, into 
the Choquet integral ${\cal C}_Q(\theta)$, which is also a Hermitian positive semidefinite operator.
The formalism uses capacities (non-additive probabilities)  and Choquet integrals, and we briefly introduce these concepts.

\paragraph*{Capacities (non-additive probabilities):}
The basic property of Kolmogorov probabilities is additivity ($\mu(A\cup B)-\mu(A)-\mu(B)+\mu(A\cap B)=0$). 
But in subjects like Artificial Intelligence, Operations Research, Game Theory, Mathematical Economics, etc,
nonadditive probabilities have been used extensively (e.g., \cite{D1,D2,D3,D4,D5}).
They are particularly useful in problems where the various alternatives overlap, and
they formalize the added value in an aggregation, where the `whole is greater than the sum of its parts'.

In recent work\cite{VO1} we have shown that there is a strong link between the non-commutativity of general projectors,
and the non-additivity of the corresponding probabilities (Eq.(\ref{e3}) below). This leads naturally to Choquet integrals, which we introduce 
in this paper in a quantum context, and discuss their use as figures of merit 
of the quantum properties of Hermitian operators.

\paragraph*{Choquet integrals in a classical context:}
Integration is based on additivity. Integrals with non-additive probabilities require another approach, and this leads to Choquet integration\cite{INT1},
which has been used extensively in Artificial Intelligence\cite{INT2,INT3,INT4,INT5,INT6}, 
in Game Theory and its applications in Mathematical Economics\cite{W1,W2,W3,W4}, etc. 
In a `weighted average' we have a number of independent alternatives, and we assign a probability to each alternative.
The Choquet integral is a `sophisticated weighted average', for non-additive probabilities related to overlapping alternatives.
It replaces the probability distributions used in weighted averages,
with the derivative of cumulative functions, and by doing so, 
it assigns weights to aggregations of alternatives.
The weight for an aggregation of alternatives, is in general different from the sum of the weights, of the alternatives it contains.
Consequently, the derivative of cumulative functions is in general different from the probability distributions (they are always equal in the case of additive probabilities).

\paragraph*{Choquet integrals in a quantum context:}
In this paper we map the $Q$ function of a Hermitian 
positive semidefinite operator $\theta $ (defined with respect to the set $\Omega$ of coherent states), into  the discrete Choquet integral
${\cal C}_Q(\theta)$.
It is calculated using cumulative projectors, and their discrete derivatives (differences) which are $d$ orthogonal projectors.

An important concept related to Choquet integrals, is comonotonicity of two operators $\theta, \phi$. There is added value
in an aggregation of components with different properties,  
because the various components play complementary role to each other.
In this case the whole is different from the sum of its parts, and the ${\cal C}_Q(\theta +\phi)$ is different from ${\cal C}_Q(\theta)+{\cal C}_Q(\phi)$.
But if the components of an aggregation have similar properties, this complementarity and added value are missing, the whole is equal to the sum of its parts, and
${\cal C}_Q(\theta +\phi)$ is equal to ${\cal C}_Q(\theta)+{\cal C}_Q(\phi)$.
Comonotonicity defines rigorously the intuitive concept of physically similar operators.

We next compare briefly the Choquet formalism with the spectral formalism of 
orthogonal projectors, the positive operator valued measures (POVM), and the formalism of frames and wavelets:
\begin{itemize}
\item
The spectral formalism of eigenvalues and eigenvectors, leads  in the case of Hermitian operators to orthogonal projectors,
which play a fundamental role in von Neumann's measurement theory.
\item
The POVM formalism uses projectors related to coherent states (or other non-orthogonal and non-commuting projectors), and it is based on a resolution of the identity.
The resolution of the identity is crucial for the calculation of various physical quantities 
in terms of coherent states. 
\item
The formalism of frames and wavelets, is based on lower and upper bounds to a resolution of the identity,
and in this sense it uses approximate resolutions of the identity with bounded error.
\item
The Choquet formalism uses a `weak resolution of the identity', that involves not only the 
non-orthogonal projectors, but also a correction which consists of `M\"obius operators' that eliminate the `double counting' in the sum 
of the non-orthogonal projectors (Eq.(\ref{al1})).
The Choquet integral ${\cal C}_{Q}(\theta)$ can be expressed in terms of the M\"obius operators (as in proposition \ref{www}).

\end{itemize}

\paragraph*{Physical applications:}
The Choquet integral ${\cal C}_Q(\theta)$ is a figure of merit of the quantum properties of Hermitian operators.
Its physical applications include:
\begin{itemize}
\item
upper and lower bounds to various physical quantities in terms of the $Q$-function (proposition \ref{1234}).
This includes the derivation of bounds for partition functions (section \ref{partition}).

\item
the study of changes in the ground state of physical systems.
Hamiltonians with and without degeneracies are considered, and it
is shown that the Choquet integral ${\cal C}_Q(\theta)$ detects changes in the ground state of the system (sections \ref{ex12}, \ref{DD}).

\item
the formalism leads naturally to the concept of comonotonicity.
It is  shown that comonotonic operators are bounded within certain intervals, 
with respect to an order based on Choquet integrals,
and in this sense they are similar to each other (section \ref{L}).
\end{itemize}
A desirable feature of the formalism, is that it is robust in the presence of noise, and yet it is sensitive enough to detect changes in the 
physical system (e.g., changes in the ground state of the system).

\paragraph*{Contents:}
In section 2, we introduce capacities and Choquet integrals in a classical context.
There is much literature on these concepts in other than Physics areas,
and here we present briefly the concepts that we are going to bring into Quantum Physics.
In section 3, we introduce technical details (cumulative coherent projectors and their discrete derivatives,
M\"obius operators, etc \cite{v16}) which are needed in the calculation of the Choquet integral.

In section 4, we introduce the Choquet integral ${\cal C}_Q(\theta)$ of a Hermitian operator $\theta$, and study its properties.
In section 5, we discuss the concept of comonotonic operators.
In section 6, we introduce an order based on Choquet integrals,  and show that comonotonic operators are bounded 
with respect to this order. This implies that certain physical quantities are also bounded.
In section 7, we apply the formalism to the study of the ground state of a physical system.
In section 8, we derive bounds for partition functions.
In section 9 we compare and contrast the Choquet formalism, with the spectral formalism of eigenvectors/eigenvalues,
the POVM formalism, and the formalism of wavelets (and frames).
We conclude in section 10, with a discussion of our results.

\section{Capacities and discrete Choquet integrals in a classical context}
\subsection{Capacities for overlapping and non-independent alternatives} 
Kolmogorov probability is a map $\mu$ from subsets of a `set of alternatives' $\Omega$, to $[0,1]$.
Its basic property is the additivity relation
\begin{eqnarray}\label{1}
\delta (A,B)=0;\;\;\;\;\;\delta (A,B)=\mu(A\cup B)-\mu(A)-\mu(B)+\mu(A\cap B);\;\;\;A,B \subseteq \Omega.
\end{eqnarray}
In the case $A\cap B=\emptyset$ this reduces to 
\begin{eqnarray}\label{B}
A\cap B=\emptyset\;\;\rightarrow\;\;\mu (A\cup B)= \mu (A)+\mu (B)
\end{eqnarray}
Capacity or nonadditive probability, is a weaker concept which obeys the relations
\begin{eqnarray}\label{A}
&&\mu (\emptyset )=0;\;\;\;\;\mu (\Omega)=1\nonumber\\
&&A \subseteq B\;\rightarrow\;\mu (A)\le \mu (B)
\end{eqnarray}
If the second of these requirements is replaced with the additivity relation of Eq.(\ref{B})
which is stronger, then the capacity is probability.
A prerequisite for the use of probabilities is the assumption that the alternatives in the set $\Omega$
are separable from each other, and truly independent.
In capacities this assumption is relaxed, the aggregation of some of the alternatives is different from the sum of its parts, and 
Eq.(\ref{B}) is not valid.

Capacities have been introduced by Choquet \cite{INT1}, and they have been used extensively in areas like
Artificial Intelligence, Operations Research, Game Theory, Mathematical Economics, etc.
They describe the added value in an aggregation, where the `whole is greater than the sum of its parts'.
For example, the percentage of votes in a coalition of two political parties, might be greater (or smaller) 
than the sum of the percentages in the component parties. 
For capacities the $\delta (A,B)$ can be positive or negative, in which case we say that the capacities 
are supermodular or submodular.

\begin{remark}
In a quantum context the requirement for probabilities, of independent alternatives in $\Omega$, corresponds to the use of an orthonormal basis.
In the case of coherent states, the non-validity of the analogue of Eq.(\ref{B}), is given with  the operator in Eq.(\ref{32}) below, which is non-zero.
\end{remark}

\subsection{ Ranking and derivatives of cumulative functions in Choquet integrals}

The nonadditivity in capacities implies that the concept of integration needs revision.
The Choquet integration is appropriate in this case. 
We consider a function $f$ on the finite set $\Omega$, which takes the real values $f(1),...,f(N)$.
We note that Choquet integrals can  also be defined for functions with a continuum of real values, but in this paper we consider the finite case.
We relabel this function, using a `ranking permutation' $i=\sigma(j)$ of the indices, so that 
\begin{eqnarray}\label{234}
f[\sigma (1)]\le ...\le f[\sigma (N)].
\end{eqnarray}
The Choquet integral of $f$ with respect to the capacity $\mu$ is given by
\begin{eqnarray}\label{87}
{\cal C}(f;\mu)&=&\sum _{i=1}^{N} f[\sigma (i)]\nu _f(i)\nonumber\\
\nu _f(i)&=&\mu(\sigma (i), \sigma (i+1),...,\sigma (N))-\mu(\sigma (i+1), \sigma (i+2),...,\sigma (N));\;\;\;\;i=1,...,N-1\nonumber\\
\nu _f(N)&=&\mu(\sigma (N));\;\;\;\;\sum_{i=1}^N\nu _f(i)=1.
\end{eqnarray}
The $\mu(\sigma (1), \sigma (2),...,\sigma (i-1))$ is a cumulative function, and 
\begin{eqnarray}\label{A6}
\mu(\sigma (i), \sigma (i+1),...,\sigma (N))=1-\mu(\sigma (1), \sigma (2),...,\sigma (i-1))
\end{eqnarray}
is a complementary cumulative function.
The $\nu _f(i)$ can be viewed as `discrete derivative' of the cumulative function.
For additive capacities (additive probabilities) the derivative of the cumulative function is equal to the probability distribution:
\begin{eqnarray}
\mu(\sigma (i), \sigma (i+1),...,\sigma (N))-\mu(\sigma (i+1), \sigma (i+2),...,\sigma (N))=\mu(\sigma (i));\;\;\;\;i=1,...,N-1
\end{eqnarray}
but for non-additive capacities this is not true (in general $\nu_f(i)\ne \mu(\sigma (i))$).
The $\nu_f(i)$ depends on $\sigma (i+1),...,\sigma (N)$ and therefore it depends on the ranking in Eq.(\ref{234}).
This is indicated in the notation $\nu_f(i)$ with the index $f$.
For two functions $f,g$, in general $\sigma_f(j)\ne \sigma_g(j)$, and therefore $\nu _f(i)\ne \nu _g(i)$.

In a weighted average we multiply the values of a function with the corresponding probabilities.
In a Choquet integral we replace the probabilities with discrete derivatives (differences) of cumulative functions. 
\begin{remark}\label{rema}
The ${\cal C}$ in the notation, indicates Choquet integral.
Both the ranking of the function in Eq.(\ref{234}), and also the derivatives of cumulative functions, play a crucial role in Choquet integrals.
Both of them will be linked to non-commutativity, in a quantum context later.
\end{remark}
If $\mu_1$ and $\mu _2$ are capacities, then $\mu=p\mu _1+(1-p)\mu _2$ where $0\le p \le 1$ is also a capacity, and
\begin{eqnarray}\label{mm1}
{\cal C}(f;\mu)=p{\cal C}(f;\mu _1)+(1-p){\cal C}(f;\mu _2) 
\end{eqnarray}
\begin{example}
If $A\subseteq \Omega$ and $\mu$ is a capacity such that 
\begin{eqnarray}\label{cv1}
&&\mu (B)=1\;\;{\rm if}\;\;A\subseteq B\subseteq \Omega\nonumber\\
&&\mu (B)=0\;\;{\rm otherwise},
\end{eqnarray}
then Eq.(\ref{87}) reduces to
\begin{eqnarray}\label{mm2}
{\cal C}(f;\mu)=\min [f(A)],
\end{eqnarray}
where ${\min}[f(A)]$ is the minimum of all the values of the function $f$ in the subset $A$.
Here only the aggregations of alternatives in sets $B\supseteq A$ make a contribution to the integral, because for any other set
$C$, we have $\mu (C)=0$.
In this case, the integral is equal to the minimum of the values of the function $f$, in the set $A$.
The contribution of the other values of the function $f$ is zero, because the discrete derivatives of the 
corresponding cumulative function $\mu(\{\sigma (i), \sigma (i+1),...,\sigma (N)\})$,
are zero.
\end{example}

\subsection{M\"obius transform: how to avoid double-counting}

The M\"obius transform is used extensively in Combinatorics, 
after the work by Rota\cite{R,R1}. It is a generalization of
the inclusion-exclusion principle that gives the cardinality of the union of overlaping sets.
The M\"obius transform describes the overlaps between sets, and it is used to avoid the `double-counting'.
Rota generalized this to partially ordered structures.

The M\"obius transform of capacities leads to the following function ${\mathfrak d} (A)$ (where $A\subseteq \Omega$):
\begin{eqnarray}\label{M}
{\mathfrak d} (A)=\sum  _{B\subseteq A}(-1)^{|A|-|B|} \mu(B).
\end{eqnarray}
where $|A|$, $|B|$ are the cardinalities of these sets.
For example, if $A=\{a_1,...,a_m\}$, then
\begin{eqnarray}
{\mathfrak d}(a_i)&=&\mu (a_i);\;\;\;\;\;
{\mathfrak d} (a_i,a_j)=\mu (a_i,a_j)-\mu(a_i)-\mu(a_j)\nonumber\\
{\mathfrak d} (a_i,a_j,a_k)&=&\mu (a_i,a_j,a_k)-\mu(a_i,a_j)-\mu(a_i,a_k)-\mu(a_j,a_k)\nonumber\\&+&\mu(a_i)+
\mu(a_j)+\mu(a_k)
\end{eqnarray}
etc.

The inverse M\"obius transform is the intuitively nice relation
\begin{eqnarray}\label{b7}
\mu (A)=\sum _{B\subseteq A}{\mathfrak d} (B).
\end{eqnarray}
For sets with one element only, ${\mathfrak d} (A)=\mu(A)$.
If $\Omega=\{a_1,...,a_n\}$, Eq.(\ref{b7}) with $A=\Omega$ becomes
\begin{eqnarray}\label{b70}
\sum _{i=1}^n\mu (a_i)+\sum _{i,j}{\mathfrak d} (a_i,a_j)+...+{\mathfrak d} (a_1,...,a_n)=1.
\end{eqnarray}
There are $2^n-1$ terms in this sum (there are $2^n$ subsets of $\Omega$, but we exclude the empty set).
Eq.(\ref{b70}) is important because lack of additivity means that in general 
\begin{eqnarray}
\sum _{i=1}^n\mu (a_i)\ne 1.
\end{eqnarray}
In the special case that the capacity $\mu (A)$ is additive (i.e., Eq.(\ref{B}) holds), ${\mathfrak d} (B)$ is zero if the cardinality of $B$ is greater or equal to $2$: 
\begin{eqnarray}
|B|\ge 2\;\;\rightarrow\;\;{\mathfrak d} (B)=0.
\end{eqnarray}

\begin{remark}
In a quantum context the analogue of the M\"obius transforms ${\mathfrak d}$, are the operators in Eq.(\ref{m}).
\end{remark}
\begin{lemma}
The Choquet integral of Eq.(\ref{87}) is given in terms of the M\"obius transform ${\mathfrak d}$ of the capacities $\mu$, as
\begin{eqnarray}\label{88}
{\cal C}(f; \mu)&=&\sum _{A\subseteq \Omega} {\mathfrak d} (A){\min}[f(A)].
\end{eqnarray}
For all subsets $A$ of $\Omega$, we multiply ${\mathfrak d} (A)$ with the minimum value of the function in this subset, and we add the results.
\end{lemma}
\begin{proof}
Explicit proof is given in \cite{INT2,INT3,INT4,INT5,INT6}.
Here we only give a hint of the proof, which is based on Eqs.(\ref{mm1}),(\ref{mm2}).
In the special case of the capacity in Eq.(\ref{cv1}), we get
\begin{eqnarray}
&&{\mathfrak d}(B)=1\;\;{\rm if}\;\;B=A\nonumber\\
&&{\mathfrak d}(B)=0\;\;{\rm otherwise}
\end{eqnarray}
and Eq.(\ref{88}), reduces to Eq.(\ref{mm2}).
In this sense, Eq.(\ref{88}), is simply a generalization of Eq.(\ref{mm2}).
\end{proof}
\begin{proposition}
The Choquet integral of Eq.(\ref{87}) can be written as the sum 
\begin{eqnarray}\label{88a}
&&{\cal C}(f; \mu)={\cal C}_1(f; \mu)+{\cal C}_2(f; \mu)+...+{\cal C}_N(f; \mu)\nonumber\\
&&{\cal C}_1(f; \mu)=\sum  {\mathfrak d} (a_i)f(a_i)=\sum  \mu (a_i)f(a_i)\nonumber\\
&&{\cal C}_2(f; \mu)=\sum {\mathfrak d} (a_i,a_j){\min}[f(a_i),f(a_j)]\nonumber\\
&&..................................................\nonumber\\
&&{\cal C}_N(f;\mu)={\mathfrak d} (a_1,...,a_M){\min}[f(a_1),...,f(a_M)].
\end{eqnarray}
In the special case of additive capacities (Kolmogorov probabilities) ${\cal C}_2(f; \mu)=...={\cal C}_N(f; \mu)=0$, and the 
Choquet integral is the standard weighted average ${\cal C}_1(f; \mu)$.
\end{proposition}
\begin{proof}
We start from Eq.(\ref{88}), and we group together all terms with ${\mathfrak d} (A)$ where the cardinality of $A$ is equal to $k$.
This gives the ${\cal C}_k(f; {\mathfrak d})$.
In the special case of additive capacities, 
all the ${\mathfrak d} (a_{i_1},...,a_{i_r})$ with $r\ge 2$ are zero, and consequently all the ${\cal C}_k(f; \mu)$ with $k\ge 2$ are equal to zero.
\end{proof}
The Choquet integral is the weighted average given by ${\cal C}_1(f; \mu)$, plus the corrections of the other terms
which are due to deviations from the additivity of probability.
In a quantum context later, these extra corrections are due to higher order M\"obius operators and are directly related to 
the overlaps between the coherent states.

\subsection{Comonotonic functions and the weak additivity property of Choquet integrals}\label{nnmm}
Ranking is important in Choquet integrals. In general two functions have different ranking and  
\begin{eqnarray}
{\cal C}(f+g;\mu)\ne {\cal C}(f;\mu)+{\cal C}(g;\mu).
\end{eqnarray}
We define comonotonic (same ranking) functions\cite{INT1,INT2,INT3,INT4,INT5,INT6}, as follows:
\begin{definition}
Two functions $f(i)$ and $g(i)$ on the set $\Omega$, are called comonotonic if the following statements which are equivalent to each other, hold:
\begin{itemize}
\item[(1)]
The ranking permutation $\sigma (j)$ is the same for the two functions ($\sigma _f(j)=\sigma _ g(j)$).
Therefore $\nu_f(i)=\nu _g(i)$. The weights $\nu_f(i)$ in Eq.(\ref{87}) are the same for comonotonic functions, but they are in general different for non-comonotonic functions. 
\item[(2)]
For all $i,j$ 
\begin{eqnarray}\label{c4}
[f(i)-f(j)][g(i)-g(j)]\ge 0.
\end{eqnarray}
\end{itemize}
\end{definition}
Comonotonicity does not obey the transitivity property (i.e., if $f_1,f_2$ are comonotonic, and $f_2,f_3$ are also comonotonic, the $f_1,f_3$ might not be comonotonic).
\begin{proposition}
\mbox{}
\begin{itemize}
\item[(1)]
For comonotonic functions $f,g$ and positive $a,b$
\begin{eqnarray}\label{301}
{\cal C}(af+bg;\mu)= a{\cal C}(f;\mu)+b{\cal C}(g;\mu).
\end{eqnarray}
\item[(2)]
If $\mu$ is an additive capacity (Kolmogorov probability), then for any functions $f,g$
\begin{eqnarray}\label{305}
{\cal C}(af+bg;\mu)= a{\cal C}(f;\mu)+b{\cal C}(g;\mu).
\end{eqnarray}
\end{itemize}
\end{proposition}
\begin{proof}
\mbox{}
\begin{itemize}
\item[(1)] 
Comonotonic functions have the same weights $\nu _f(i)$ in Eq.(\ref{87}), and this proves Eq.(\ref{301}).
The $a,b$ are taken to be positive, so that the ranking is preserved.
\item[(2)]
For additive capacities, 
 all the ${\mathfrak d} (a_{i_1},...,a_{i_r})$ with $r\ge 2$ in Eq.(\ref{88a}), are zero.
Therefore only the ${\cal C}_1(f; \mu)$ is non-zero, and then we can easily prove Eq.(\ref{305}).
\end{itemize}

\end{proof}
For general (non-additive) capacities, additivity of the Choquet integral holds only for comonotonic functions.
We refer to this as weak additivity property of Choquet integrals.

A constant function $c$ is comonotonic with any function $f$ and therefore
\begin{eqnarray}
{\cal C}(f+c;\mu)={\cal C}(f;\mu)+c.
\end{eqnarray}

\subsection{Example}\label{aatt}

Four students $A,B,C,D$ were examined in three modules $1,2,3$ and they got the following marks (in the interval $[0,100]$):
\begin{eqnarray}\label{12}
&&f_A(1)=70;\;\;\;\;\;f_A(2)=70;\;\;\;\;\;f_A(3)=30\nonumber\\
&&f_B(1)=90;\;\;\;\;\;f_B(2)=50;\;\;\;\;\;f_B(3)=80\nonumber\\
&&f_C(1)=50;\;\;\;\;\;f_C(2)=90;\;\;\;\;\;f_C(3)=70\nonumber\\
&&f_D(1)=70;\;\;\;\;\;f_D(2)=60;\;\;\;\;\;f_D(3)=50.
\end{eqnarray}
A professor considers them as applicants for a PhD study, taking into account how close the three modules are to
the topic of the Ph.D.
The assumption of separability and independence of the three modules is too strong (because usually the modules overlap with each other).
We adopt the weaker concepts of capacity and Choquet integrals, which allow for an aggregation 
to be different from the sum of its parts.

In this example, $\Omega$ is the set of the three modules $\{1,2,3\}$.
We will calculate the Choquet integrals using the capacities:
\begin{eqnarray}\label{350}
&&\mu(1)=0.3;\;\;\;\;\;\mu(2)=0.3;\;\;\;\;\;\mu(3)=0.2\nonumber\\
&&\mu(1,2)=1;\;\;\;\;\;\mu(1,3)=0.4;\;\;\;\;\;\mu(2,3)=0.4\nonumber\\
&&\mu(\emptyset )=0;\;\;\;\;\;\mu(1,2,3)=1
\end{eqnarray}
It is not a requirement that the $\mu(1)+\mu(2)+\mu(3)$ should be equal to $1$ (see Eq.(\ref{b70})).
These capacities reflect the fact that the aggregation of modules $1$ and $2$ is ideal for 
the topic of this PhD, and for this reason $\mu(1,2)>\mu(1)+\mu(2)$ (and in fact $\mu(1,2)=1$).
The aggregation of modules $1$ and $3$ is not very important for this Ph.D., and this is reflected in the $\mu(1,3)<\mu(1)+\mu(3)$.
One reason for this may be that there is overlap in the material taught in modules $1,3$.
Similar comment can be made for the aggregation of modules $2,3$.
In the quantum context later, the reason for the non-additivity is the non-zero overlap between coherent states.

For student $A$ we have $f_A(3)\le f_A(2)\le f_A(1)$ and therefore
\begin{eqnarray}\label{zxc}
&&\nu_A(3)=1-\mu(1,2)=0;\;\;\;\;\nu_A(2)=\mu(1,2)-\mu(1)=0.7;\;\;\;\;\nu_A(1)=\mu(1)=0.3\nonumber\\
&&{\cal C}_A(f;\mu)=f_A(3)\nu_A(3)+f_A(2)\nu_A(2)+f_A(1)\nu_A(1)=70.
\end{eqnarray}
Since $\nu_A(3)=0$ the lowest mark of this student does not contribute in the calculation.
For student $B$ we have $f_B(2)\le f_B(3)\le f_B(1)$ and therefore
\begin{eqnarray}
&&\nu_B(2)=1-\mu(1,3)=0.6;\;\;\;\;\nu_B(3)=\mu(1,3)-\mu(1)=0.1;\;\;\;\;\nu_B(1)=\mu(1)=0.3\nonumber\\
&&{\cal C}_B(f;\mu)=f_B(2)\nu_B(2)+f_B(3)\nu_B(3)+f_B(1)\nu_B(1)=65.
\end{eqnarray}
For student $C$ we have $f_C(1)\le f_C(3)\le f_C(2)$ and therefore
\begin{eqnarray}
&&\nu_C(1)=1-\mu(2,3)=0.6;\;\;\;\;\;\nu_C(3)=\mu(2,3)-\mu(2)=0.1;\;\;\;\;\;\;\nu_C(2)=\mu(2)=0.3\nonumber\\
&&{\cal C}_C(f;\mu)=f_C(1)\nu_C(1)+f_C(3)\nu_C(3)+f_C(2)\nu_C(2)=64.
\end{eqnarray}
The marks of the student $D$ are comonotonic to those of the student $A$.
This means that the students $A,D$ have similar academic strengths and weaknesses, with respect to the modules $\{1,2,3\}$
(the analogue of this in a quantum context will be physically similar Hermitian operators).
Therefore $\nu_D(3)=\nu_A(3)=0$ and $\nu_D(2)=\nu_A(2)=0.7$ and $\nu_D(1)=\nu_A(1)=0.3$.
It follows that
\begin{eqnarray}
{\cal C}_D(f;\mu)=f_D(3)\nu_D(3)+f_D(2)\nu_D(2)+f_D(1)\nu_D(1)=63.
\end{eqnarray}
The Choquet integral is a figure of merit, which orders the students as $A \succ B\succ C \succ D$.
Here $A \succ B$ means that A is more (or equally) preferable for Ph.D. than B.  

We note that the weight of the same subject is different for different students.
For example, $\nu _A(2)=0.5$, $\nu _B(2)=0.6$, $\nu _C(2)=0.3$.
This is related to the fact that the three modules are not independent.
Cumulative rather than separable weights are used in the calculation.
The ranking in Eq.(\ref{234}) plays an important role in determining the 
values of $\nu$.

The M\"obius transform of the capacities in Eq.(\ref{350}) gives
\begin{eqnarray}
&&{\mathfrak d} (1,2,3)=0;\;\;\;\;{\mathfrak d} (1,2)=0.4;\;\;\;\;{\mathfrak d} (1,3)=-0.1;\;\;\;\;{\mathfrak d} (2,3)=-0.1\nonumber\\
&&{\mathfrak d} (1)=0.3;\;\;\;\;{\mathfrak d} (2)=0.3;\;\;\;\;{\mathfrak d} (3)=0.2
\end{eqnarray}
Then using Eq.(\ref{88}) we find the same results as above. We present explicitly the calculation for one of them.
Taking into account that
\begin{eqnarray}
\min \{f_A(1),f_A(2)\}=f_A(2);\;\;\;\;\;\min \{f_A(1),f_A(3)\}=f_A(3);\;\;\;\;\;\min \{f_A(2),f_A(3)\}=f_A(3)
\end{eqnarray}
we get
\begin{eqnarray}
{\cal C}_A(f;{\mathfrak d})&=&{\cal C}_A^{(1)}(f;{\mathfrak d})+{\cal C}_A^{(2)}(f;{\mathfrak d})+{\cal C}_A^{(3)}(f;{\mathfrak d})=70\nonumber\\
{\cal C}_A^{(1)}(f;{\mathfrak d})&=&{\mathfrak d} (1)f_A(1)+{\mathfrak d} (2)f_A(2)+{\mathfrak d} (3)f_A(3)=48\nonumber\\
{\cal C}_A^{(2)}(f;{\mathfrak d})&=&{\mathfrak d} (1,2)f_A(2)+{\mathfrak d} (1,3)f_A(3)+{\mathfrak d} (2,3)f_A(3)=22\nonumber\\
{\cal C}_A^{(3)}(f;{\mathfrak d})&=&{\mathfrak d} (1,2,3)f_A(3)=0,
\end{eqnarray}
which is the same result as in Eq.(\ref{zxc}).

We note that if we use the `standard average' we find
\begin{eqnarray}
{\mathfrak M}_A=\frac{170}{3};\;\;\;\;{\mathfrak M}_B=\frac{220}{3};\;\;\;\;{\mathfrak M}_C=\frac{210}{3};\;\;\;\;{\mathfrak M}_D=\frac{180}{3}.
\end{eqnarray}
and this leads to the ordering $B\sqsupset C \sqsupset D \sqsupset A$, where $\sqsupset$ is the ordering according to the `standard averaging'.

\section{Cumulative projectors and M\"obius operators}

\subsection{Coherent states}

We consider a quantum system $\Sigma (d)$ with variables in ${\mathbb Z}(d)$, and $d$-dimensional Hilbert space $H(d)$.
We also consider the orthonormal basis of `position states' $\ket{X;n}$,
and through the Fourier transform $F$, the basis of momentum states $\ket{P;n}$\cite{Fin,Fin2}:
\begin{eqnarray}
&&F=d^{-1/2}\sum _m\omega (mn)\ket{X;n}\bra{X:m};\;\;\;\;\omega(\alpha )=\exp \left(\frac{i2\pi \alpha}{d}\right )
\nonumber\\
&&\ket{P;n}=F\ket{X;n};\;\;\;\;m,n, \alpha\in {\mathbb Z}(d).
\end{eqnarray}
Displacement operators in the ${\mathbb Z}(d)\times {\mathbb Z}(d)$ phase space, are given by
\begin{eqnarray}\label{dis}
D(\alpha , \beta)=Z^{\alpha}X^{\beta}\omega (-2^{-1}\alpha \beta);\;\;\;\;
Z=\sum _m\omega (m)\ket{X;m}\bra{X;m};\;\;\;\;
X=\sum _m \ket{X;m+1}\bra{X;m}
\end{eqnarray}
The $\{D(\alpha , \beta)\omega (\gamma)\}$ form the Heisenberg-Weyl group of displacements in this system.
The formalism of finite quantum systems, is slightly different in the cases of odd and even $d$.
The factor $2^{-1}$ above, is an element of ${\mathbb Z}(d)$, and it exists only for odd $d$. 
Below we assume that the dimension $d$ is an odd integer. 

Acting with $D(\alpha , \beta)$ on a (normalized) fiducial vector $\ket {\eta}$,
we get the $d^2$ coherent states\cite{COH,COH1}:
\begin{eqnarray}\label{coh}
\ket{C;\alpha, \beta}=D(\alpha , \beta)\ket{\eta};\;\;\;\;\ket {\eta}=\sum _m \eta _m\ket{X;m};\;\;\;\;\sum _m|\eta_m|^2=1.
\end{eqnarray}
The $X,P,C$ in the notation are not variables, but they simply indicate position states, momentum states and coherent states.
We call $\Omega$ the set of the $d^2$ coherent states:
\begin{eqnarray}\label{33}
\Omega =\{\ket{C;\alpha, \beta}\;|\;\alpha , \beta \in {\mathbb Z}(d)\}.
\end{eqnarray}
The set $\Omega$ is invariant under displacement transformations.

Let $\Pi(\alpha,\beta)$ be the projector to the one-dimensional subspace $H(\alpha,\beta)$ that contains the coherent states $\ket{C;\alpha,\beta}$. Then
\begin{eqnarray}\label{1111}
&&\frac{1}{d}\sum _{\alpha,\beta}\Pi({\alpha,\beta})={\bf 1};\;\;\;\;\;\;\Pi ({\alpha,\beta})=\ket{C;\alpha,\beta}\bra{C;\alpha,\beta}\nonumber\\
&&D(\gamma, \delta)\Pi({\alpha ,\beta })D^{\dagger}(\gamma, \delta)=\Pi({\alpha +\gamma,\beta +\delta})
\end{eqnarray}
The term `coherent states' refers to these two properties. They are the analogue of the harmonic oscillator coherent states \cite{coh1,coh2,coh3}, in the context of quantum systems with finite-dimensional Hilbert space.

Let ${\cal M}_d$ be the set of $d\times d$ Hermitian positive semidefinite matrices, and ${\cal N}_d\subset {\cal M}_d$ the set of $d\times d$ density matrices.
For $\theta \in {\cal M}_d$, the $Q$-function is given by
\begin{eqnarray}\label{mmm}
Q(\alpha,\beta\;|\;\theta)=\frac{1}{d}{\rm Tr}[\Pi (\alpha,\beta)\theta];\;\;\;\;\;\;\sum _{\alpha, \beta} Q(\alpha,\beta\;|\;\theta)={\rm Tr} \theta,
\end{eqnarray}
and the $P$-function by
\begin{eqnarray}\label{9}
\theta=\sum _{\alpha,\beta}P(\alpha,\beta\;|\;\theta)\Pi ({\alpha,\beta});\;\;\;\;\;
\sum _{\alpha,\beta}P(\alpha,\beta\;|\;\theta)={\rm Tr}\theta.
\end{eqnarray}
If $\theta _{mn}=\bra{X;m}\theta\ket{X;n}$
then
\begin{eqnarray}\label{xxx}
Q(\alpha,\beta\;|\;\theta)=\sum \theta _{mn}A(m,n;\alpha, \beta);\;\;\;\;A(m,n;\alpha, \beta)=\frac{1}{d}\bra{C;\alpha, \beta}X;m\rangle \langle X;n\ket{C;\alpha, \beta}
\end{eqnarray}
The $A(m,n;\alpha, \beta)$ is a $d^2\times d^2$ matrix, and the fiducial vector should be such that its determinant is non-zero.
Then Eq.(\ref{xxx}) is a system with $d^2$ equations and if the $Q(\alpha,\beta\;|\;\theta)$ are known we can calculate the $\theta _{mn}$, and vice-versa.

\paragraph*{Wehrl entropy for the $Q$-function of density matrices:}
For $\theta \in {\cal M}_d$, 
we define the $\widetilde {\theta } =\theta /{\rm Tr}\theta \in {\cal N}_d$, which can be viewed as a density matrix.
Its Wehrl entropy\cite{We} is given by
\begin{eqnarray}\label{entro}
E(\widetilde {\theta })=-\sum _{\alpha, \beta} Q(\alpha,\beta\;|\;\widetilde {\theta } )\log Q(\alpha,\beta\;|\;\widetilde {\theta } );\;\;\;\;\;\;\sum _{\alpha, \beta} Q(\alpha,\beta\;|\;\widetilde {\theta })=1.
\end{eqnarray}
Its maximum value is $d\log d$.
Under any permutation $(\gamma, \delta)=\sigma (\alpha, \beta)$ of the indices of the $Q$-function, the Wehrl entropy $E(\widetilde {\theta})$, does not change:
\begin{eqnarray}
E(\widetilde {\theta })=-\sum Q(\alpha,\beta\;|\;\widetilde {\theta } )\log Q(\alpha,\beta\;|\;\widetilde {\theta } )
=-\sum Q[\sigma(\alpha,\beta)\;|\;\widetilde {\theta } ]\log Q[\sigma(\alpha,\beta)\;|\;\widetilde {\theta } ]
\end{eqnarray}
Therefore the Wehrl entropy does not tell us, for which coherent states we get high (or low) value of the $Q$-function. The Wehrl entropy
shows whether the $Q$-function is uniform or concentrated in a few coherent states, but in the latter case it does not show where it is concentrated.
Depending on the application, this might be a desirable or undesirable property of the Wehrl entropy.
This is also seen by the fact that under displacement transformations, the Wehrl entropy does not change:
\begin{eqnarray}\label{bbb}
E\left [D(\alpha, \beta) \widetilde {\theta } D^{\dagger}(\alpha, \beta)\right ]=E(\widetilde {\theta }).
\end{eqnarray}
We stress that the $Q(\alpha,\beta\;|\;\widetilde {\theta })$ are not probabilities, because the coherent states overlap with each 
other ($d^2$ coherent states in a $d$-dimensional space).
Related to this, is that the distribution $Q(\alpha,\beta\;|\;\widetilde {\theta } )$ can not be very narrow, and consequently 
the Wehrl entropy $E(\widetilde {\theta })$ is greater than a certain value
(which in the harmonic oscillator case is equal to one\cite{LIEB}).
The motivation for introducing Choquet integrals later, is to quantify and elucidate 
the effects of these overlaps between the coherent states.

\subsection{Two-dimensional cumulative projectors}\label{eac}
We consider the two-dimensional space $H({\alpha _1,\beta _1};{\alpha _2,\beta _2})$ that contains all superpositions
$\kappa \ket {C;{\alpha _1,\beta _1}}+\lambda \ket{C;{\alpha _2,\beta _2}}$:
\begin{eqnarray}
H({\alpha _1,\beta _1};{\alpha _2,\beta _2})={\rm span}[H({\alpha _1,\beta _1}) \cup H({\alpha _2,\beta _2})]
\end{eqnarray}
In the language of lattices\cite{LO1,LO2,LO3} this is the disjunction of the one dimensional spaces $H({\alpha _1,\beta _1})$ and $H({\alpha _2,\beta _2})$.
We note that the conjuction of these spaces $H({\alpha _1,\beta _1}) \cap H({\alpha _2,\beta _2})$ contains only the zero vector.

We denote the projector to the space $H({\alpha _1,\beta _1};{\alpha _2,\beta _2})$ 
as $\Pi({\alpha _1,\beta _1};{\alpha _2,\beta _2})$ or if there is no danger of confusion simply as $\Pi(1,2)$.
The $\Pi(1,2)$ can be calculated with the Gram-Schmidt orthogonalization method,
where we take the component of $\ket {C;{\alpha _2,\beta _2}}$ which is perpendicular to $\ket {C;{\alpha _1,\beta _1}}$, and we normalize it into a vector with length $1$.
We express this in terms of projectors as
\begin{eqnarray}\label{gg}
&&\Pi (1,2)=\Pi(1)+\varpi (2|1)\nonumber\\
&&\varpi (2|1)=
\frac{\Pi^{\perp}(1)\Pi(2)\Pi^{\perp}(1)}{{\rm Tr}[\Pi^{\perp}(1)\Pi(2)]}\nonumber\\
&&\Pi^{\perp}(1)={\bf 1}-\Pi (1).
\end{eqnarray}
We call the $\Pi (1,2)$ cumulative projectors because 
they project into two-dimensional spaces, and therefore the corresponding probabilities
take a range (two) values. The $\varpi (2|1)=\Pi (1,2)-\Pi(1)$
can be viewed as a discrete derivative (difference) of the cumulative projectors. 
In additive (Kolmogorov) probabilities, the derivative of the cumulative distributions are the probability distributions.
This is not true for capacities (non-additive probabilities), precisely because additivity does not hold.
Here this is the fact that the $\varpi (2|1)$ is different from $\Pi (2)$.
From a physical point of view, a measurement with the projector $\Pi^{\perp}(1)$ 
(which projects to the orthogonal complement of $H(\alpha _{1},\beta _{1})$),
on the coherent state $\ket {C;{\alpha _2,\beta _2}}$
(which is described with the density matrix $\Pi(2)$), will collapse it into the $\varpi (2|1)$ with probability ${\rm Tr}[\Pi^{\perp}(1)\Pi(2)]$.

\begin{lemma}\label{le1}
\begin{eqnarray}\label{1019}
D(\gamma, \delta)\varpi({\alpha _2,\beta _2}|{\alpha _1,\beta _1})D^{\dagger}(\alpha , \beta)=\varpi({\alpha _2+\gamma,\beta _2+\delta}|{\alpha _1+\gamma,\beta _1+\delta})
\end{eqnarray}
\end{lemma}
\begin{proof}
We multiply both sides of the second of Eq.(\ref{gg}) by $D(\gamma, \delta)$ on the left and $D^{\dagger}(\gamma, \delta)$ on the right, taking into account Eq.(\ref{1111}).
\end{proof}

In analogy with Eq.(\ref{1})  we consider the following operator: 
\begin{eqnarray}\label{32}
{\mathfrak D} (1,2)=\Pi (1,2)-\Pi(1)-\Pi (2);\;\;\;\;{\rm Tr}[{\mathfrak D} (1,2)]=0.
\end{eqnarray} 
A projector to the space $H({\alpha _1,\beta _1}) \cap H({\alpha _2,\beta _2})$ should also 
be added to the right hand side, but as we explained earlier it is zero.
The trace of this operator with a density matrix $\rho$ converts the projectors into probabilities, and in this sense
the ${\mathfrak D} (1,2)$ is analogous to $\delta (A,B)$ in Eq.(\ref{1}). Unlike $\delta (A,B)$, the
${\mathfrak D} (1,2)$ is in general non-zero, and quantifies deviations from the additivity of probability due to the 
overlapping nature of coherent states. 
The resolution of the identity in terms of coherent states, shows that in the corresponding sum these overlaps cancel
each other. The following proposition shows that something similar happens with the  ${\mathfrak D} (1,2)$ operators:
\begin{proposition}\label{abc}
For fixed $\alpha _i, \beta _i$:
\begin{eqnarray}
\sum _{\kappa, \lambda}{\mathfrak D} ({\alpha_1+\kappa ,\beta_1+\lambda};{\alpha _2+\kappa, \beta _2+\lambda})=0.
\end{eqnarray} 
\end{proposition}
\begin{proof}
Using the resolution of the identity for coherent states, it has been proved (Eq.(119) in ref\cite{Fin2}) that for any operator $\chi$
\begin{eqnarray}
\frac{1}{d}\sum _{\kappa, \lambda}D(\kappa,\lambda)\chi [D(\kappa,\lambda)]^{\dagger}={\bf 1}{\rm Tr}\chi.
\end{eqnarray} 
We use this with $\chi=\Pi(\alpha _1, \beta _1;\alpha _2, \beta _2)$, in conjuction with the relation
\begin{eqnarray}
D(\kappa,\lambda)\Pi(\alpha _1, \beta _1;\alpha _2, \beta _2)[D(\kappa,\lambda)]^{\dagger}=
\Pi({\alpha_1+\kappa ,\beta_1+\lambda};{\alpha _2+\kappa, \beta _2+\lambda}),
\end{eqnarray}
and we prove that
\begin{eqnarray}
\frac{1}{2d}\sum _{\kappa, \lambda}\Pi({\alpha _1+\kappa ,\beta _1+\lambda};{\alpha _2+\kappa, \beta _2+\lambda})={\bf 1}.
\end{eqnarray}
This together with the resolution of the identity for $\Pi(\alpha _1, \beta _1)$ and $\Pi(\alpha _2,\beta _2)$
proves the proposition.
\end{proof}
\begin{remark}
The ${\mathfrak D} (1,2)$ are a special case of more general operators ${\mathfrak D}(H_1, H_2)$
associated with subspaces $H_1$ and $H_2$ of $H(d)$, which we have studied in \cite{VO1}.
We have proved there that
the commutator of the projectors to these subspaces
$[\Pi(H_1),\Pi(H_2)]$ is related to ${\mathfrak D}(H_1, H_2)$, through the relation:
\begin{eqnarray}\label{e3}
[\Pi (H_1),\Pi(H_2)]={\mathfrak D}(H_1, H_2)[\Pi(H_1)-\Pi(H_2)].
\end{eqnarray}
This relation links non-commutativity with non-additive probabilities.
For non-commuting projectors, the ${\rm Tr}[\rho {\mathfrak D}(H_1, H_2)]$ (where $\rho$ is a density matrix) is non-zero, and we cannot interpret the 
corresponding probabilities as additive (Kolmogorov) probabilities.
In \cite{VO1}, we interpreted quantum probabilities as non-additive (Dempster-Shafer) probabilities, for which the $\delta(A,B)$ of Eq.(\ref{1}) is in general non-zero.
\end{remark}

\subsection{Multi-dimensional cumulative projectors}

We order the coherent states in an arbitrary way and we label them as $\ket {C;{\alpha _{1},\beta _{1}}},...,\ket {C;{\alpha _{d^2},\beta _{d^2}}} $.
The formalism in this section depends on this ordering.
In the Choquet integrals, the $Q$-function of an operator $\theta$ will defne the ordering, as discussed in the next section.

We introduce inductively the space $H(\alpha _{i},\beta _{i};...;\alpha _{d^2},\beta _{d^2})$ that contains all superpositions of the $d^2-(i-1)$ coherent states
$\ket {C;{\alpha _{i},\beta _{i}}},...,\ket {C;{\alpha _{d^2},\beta _{d^2}}} $.
We start from $d^2$ and use `reverse order' because this is consistent with the ascending ordering in Eq.(\ref{234}) (and Eq.(\ref{56}) later), which is standard practice in the Choquet integrals literature.
As we go from the space $H(\alpha _{i+1},\beta _{i+1};...;\alpha _{d^2},\beta _{d^2})$ to the space $H(\alpha _{i},\beta _{i};...;\alpha _{d^2},\beta _{d^2})$,
there are two cases:
\begin{itemize}
\item
The coherent state $\ket {C;{\alpha _{i},\beta _{i}}}$ is not a linear combination of the coherent states $\ket {C;{\alpha _{i+1},\beta _{i+1}}},...,\ket {C;{\alpha _{d^2},\beta _{d^2}}}$.
The projector to the space $H(\alpha _{i},\beta _{i};...;\alpha _{d^2},\beta _{d^2})$ is
\begin{eqnarray}\label{ggg}
&&\Pi (i,...,d^2)=\Pi(i+1,...,d^2)+\varpi (i|i+1,...,d^2);\;\;\;\;\;i=1,...,d^2\nonumber\\
&&\varpi (i|i+1,...,d^2)=
\frac{\Pi^{\perp}(i+1,...,d^2)\Pi(i)\Pi^{\perp}(i+1,...,d^2)}{{\rm Tr}[\Pi^{\perp}(i+1,...,d^2)\Pi(i)]}\nonumber\\
&&\Pi^{\perp}(i,...,d^2)={\bf 1}-\Pi(i,...,d^2)\nonumber\\
&&\Pi(i+1,...,d^2)\varpi (i|i+1,...,d^2)=0.
\end{eqnarray}
The denominator in this case is different than zero, and the dimension of the space $H(\alpha _{i},\beta _{i};...;\alpha _{d^2},\beta _{d^2})$ is equal to 
the dimension of the space $H(\alpha _{i+1},\beta _{i+1};...;\alpha _{d^2},\beta _{d^2})$ plus one.
The Gram-Schmidt orthogonalization method is used here. The algorithm can also be implemented with the QR factorization of matrices \cite{matrix}, and is available in computer libraries (eg, in MATLAB).
From a physical point of view, a measurement with the projector $\Pi^{\perp}(i+1,...,d^2)$ 
(which projects to the orthogonal complement of $H(\alpha _{i+1},\beta _{i+1};...;\alpha _{d^2},\beta _{d^2})$),
on the coherent state $\ket {C;{\alpha _i,\beta _i}}$
(which is described with the density matrix $\Pi(i)$), will collapse it into the $\varpi (i|i+1,...,d^2)$ with probability 
${\rm Tr}[\Pi^{\perp}(i+1,...,d^2)\Pi(i)]$.

\item
The coherent state $\ket {C;{\alpha _{i},\beta _{i}}}$ is a linear combination of the coherent states $\ket {C;{\alpha _{i+1},\beta _{i+1}}},...,\ket {C;{\alpha _{d^2},\beta _{d^2}}}$.
In this case $\varpi (i|i+1,...,d^2)=0$ and the dimension of the space $H(\alpha _{i},\beta _{i};...;\alpha _{d^2},\beta _{d^2})$ is equal to 
the dimension of the space $H(\alpha _{i+1},\beta _{i+1};...;\alpha _{d^2},\beta _{d^2})$.
\end{itemize}

There are $d^2$ projectors $\varpi _{\theta}(i|i+1,...,d^2)$ in the $d$-dimensional space $H(d)$
(with $\varpi(d^2)=\Pi(d^2)$).
$d^2-d$ of these projectors are equal to zero, and the rest form an orthogonal and complete set of projectors in $H(d)$: 
\begin{eqnarray}\label{200}
&&\Pi (i,...,d^2)=\varpi (i|i+1,...,d^2)+...+\varpi(d^2-1|d^2)+\varpi(d^2)\nonumber\\
&&\sum _{i=1}^{d^2}\varpi (i|i+1,...,d^2)={\bf 1}.
\end{eqnarray}
Relations similar to those in lemma \ref{le1} can also be proved for the projectors $\varpi (i|i+1,...,d^2)$. 

\paragraph*{Coherent states with a generic fiducial vector:} A fiducial vector is called `generic', if any $d$ of the corresponding  coherent states are linearly independent.
In this case any set of $d$ or more coherent states is a total set in $H(d)$, i.e., there is not vector which is orthogonal to all these coherent states. 
Then for any set $A=\{i_1,...,i_d\}$ with $d$ indices 
\begin{eqnarray}\label{simple}
&&\Pi (i_1,...,i_d)={\bf 1};\;\;\;\;\;\varpi (j|i_1,...,i_d)=0;\;\;\;\;j\in\{1,...,d^2\}-A\nonumber\\
&&\varpi (i_d|i_{d-1},...,i_1)+...+\varpi(i_2|i_1)+\varpi(i_1)={\bf 1}.
\end{eqnarray}
Apart from position and momentum states, `most' of the other vectors can be used as generic fiducial  vectors.

For simplicity all our examples later, are in the $3$-dimensional space $H(3)$, and use
coherent states $D(\alpha, \beta) \ket{\eta}$ with respect to the generic fiducial vector 
\begin{eqnarray}\label{fid}
\ket{\eta }=\frac{1}{\sqrt{14}}\left (\ket{X;0}+2\ket{X;1}+3\ket{X;2}\right).
\end{eqnarray}

\subsection{M\"obius operators}\label{MMM}

The M\"obius transform of Eqs(\ref{M}),(\ref{b7}), in the present context provides a systematic method for the expression of the ${\mathfrak D}$-operators in terms of the 
cumulative $\Pi$-projectors.
If $A=\{(\alpha _1,\beta _1), (\alpha _2,\beta _2),...\}$ (where $\alpha _i,\beta _i \in {\mathbb Z}(d)$) is a set of pairs of indices,
we use the shorthand notation ${\mathfrak D}(A)$ for ${\mathfrak D}(1,2,...)$, and $\Pi(A)$ for $\Pi(1,2,...)$.
${\mathfrak D} (B)$ is related to the various projectors through the M\"obius transform \cite{R,R1}
\begin{eqnarray}\label{m}
{\mathfrak D} (B)=\sum _{A\subseteq B} (-1)^{|A|-|B|}\Pi(A).
\end{eqnarray}
For sets with only one pair ${\mathfrak D} (A)=\Pi(A)$.
A simple example of this, with two coherent states, is in Eq.(\ref{32}).
We refer to ${\mathfrak D} (B)$ as the M\"obius operators.
The trace of these operators with a density matrix, lead to probabilistic relations which quantify deviations from the additivity of probability.
The inverse M\"obius transform is
\begin{eqnarray}\label{al}
\Pi (A)=\sum _{B\subseteq A}{\mathfrak D} (B).
\end{eqnarray}
In Eq.(\ref{al}) we put $A={\mathbb Z}(d)\times {\mathbb Z}(d)$ (the set of all $(\alpha_i,\beta _i)$), and we get
\begin{eqnarray}\label{al1}
\sum _{i=1}^{d^2}\Pi(i)+
\sum _{i,j}{\mathfrak D} (i,j)+...+{\mathfrak D}(1,...,d^2)={\bf 1}.
\end{eqnarray}
This can be viewed as a kind of weak resolution of the identity, where the `M\"obius operators' eliminate the `double counting' in the sum of the non-orthogonal projectors. The term `weak'is used to indicate that in addition to the projectors, the M\"obius operators are needed.

This inverse M\"obius transform involves the $d^2$ projectors $\Pi(i)$, 
and all the M\"obius ${\mathfrak D}$-operators, whose role is to remove the overlaps between the $\Pi(i)$ so there is no double-counting.
Eq.(\ref{al1}) is the quantum analogue of Eq.(\ref{b70}).
From Eqs(\ref{1111}), (\ref{al1}) it follows that
\begin{eqnarray}
\sum _{i,j}{\mathfrak D} (i,j)+...+{\mathfrak D}(1,...,d^2)=(1-d){\bf 1}.
\end{eqnarray}
The amount of double counting in the sum $\sum \Pi(i)=d{\bf 1}$ is $(d-1){\bf 1}$, and it is cancelled by the above sum of M\"obius ${\mathfrak D}$-operators.

In the case of coherent states with a generic fiducial vector, we insert in Eq.(\ref{al}) any set with $d$ pairs of indices, 
$A=\{i_1,...,i_d\}$, and we get the following inverse M\"obius transform that involves only $d$ of the $d^2$ coherent states,
and the corresponding M\"obius operators.
\begin{eqnarray}\label{al2}
\sum _{j=1}^{d}\Pi(i_j)+
\sum _{i_j,i_k}{\mathfrak D} (i_j,i_k)+...+{\mathfrak D}(i_1,...,i_d)={\bf 1}.
\end{eqnarray}
\begin{remark}
The trace of the projectors $\Pi(A)$ times a density matrix, gives capacities.
In this sense, the projectors $\Pi(A)$ are the quantum analogue of the capacities $\mu$ in the classical case.
The $\Pi(A\cup B)\ne \Pi(A)+ \Pi(B)$ corresponds to the non-additivity of capacities.
The operators ${\mathfrak D}$ are the quantum analogue of the ${\mathfrak d}$ in the classical case.
\end{remark}

\begin{example}
In the three-dimensional space $H(3)$ we consider coherent states with a generic fiducial vector. 
For any triplet of indices $i,j,k$ (from $1,...,9$) we consider the M\"obius operators:
\begin{eqnarray}
&&{\mathfrak D} (i)=\Pi (i);\;\;\;\;{\mathfrak D} (i,j)=\Pi(i,j)-\Pi(i)-\Pi(j)\nonumber\\
&&{\mathfrak D} (i,j,k)={\bf 1}-\Pi(i,j)-\Pi(i,k)-\Pi(j,k)+\Pi(i)+\Pi(j)+\Pi(k)
\end{eqnarray}
If $A$ is a set with three of the indices $1,...,9$, then
\begin{eqnarray}
\sum _i\Pi(i)+\sum _{i,j}{\mathfrak D} (i,j)+{\mathfrak D} (i,j,k)={\bf 1};\;\;\;\;\;i,j,k\in A.
\end{eqnarray}
This involves $3$ (from the total of $9$) coherent states, and the corresponding M\"obius operators.
\end{example}

\section{The discrete Choquet integral for the $Q$-function}\label{BBB}

The formalism below is presented with the $Q$-function of operators $\theta \in {\cal M}_d$, but it can also be used with the $P$-function,
for operators $\theta$ with non-negative $P$-function.
The formalism can be extended to the more general case where $Q$ and $P$ take all real values (i.e., all Hermitian operators), but we do not discuss this in the present paper.

We relabel the $Q(\alpha,\beta\;|\;\theta)$ as $Q(i\;|\;\theta)$ ($i=1,...,d^2$) so that
\begin{eqnarray}\label{56}
0\le Q (1\;|\;\theta)\le Q(2\;|\;\theta)\le...\le Q(d^2\;|\;\theta).
\end{eqnarray}
We use here a ranking permutation 
\begin{eqnarray}\label{a56}
i=\sigma(\alpha,\beta \;|\;\theta),
\end{eqnarray}
of the $d^2$ indices $(\alpha,\beta)\in {\mathbb Z}(d)\times {\mathbb Z}(d)$
which depends on the operator $\theta$. Accordingly, we relabel the subspaces $H(\alpha,\beta)$ as
$H_{\theta}(i)$, and the projectors $\Pi(\alpha,\beta)$ as $\Pi_{\theta}(i)$. 
The index $\theta$ indicates that the labelling depends on $\theta$ (on the reordering in Eq.(\ref{a56})).
We note here that for large $d$, the ordering of the $Q(\alpha, \beta|\theta)$ can be a practically difficult problem, but there are computer programmes which do this (e.g., in MATLAB).

In analogy to the classical case in Eq.(\ref{87}), we introduce the Choquet integral ${\cal C}_Q(\theta)$ of the $Q$-function of $\theta$, as
\begin{eqnarray}\label{010}
{\cal C}_Q(\theta)&=&\sum_{i=1}^{d^2}dQ(i\;|\;\theta)\varpi _{\theta}(i|i+1,...,d^2),
\end{eqnarray}
where $\varpi _{\theta}(i|i+1,...,d^2)$ are the discrete derivatives (differences) of the cumulative projectors
\begin{eqnarray}\label{4fk}
\varpi _{\theta}(i|i+1,...,d^2)&=&\Pi _{\theta}(i;{i+1};...;{d^2})-\Pi _{\theta}({i+1};{i+2};...; {d^2})
\end{eqnarray}
These projectors are the same as Eq.(\ref{ggg}), but here the labelling depends on the ranking of the $Q$-function of $\theta$.
Since $d^2-d$ of the projectors $\varpi _{\theta}(i|i+1,...,d^2)$ are zero, it follows that only $d$ of the $d^2$ values of $Q(i\;|\;\theta)$, 
contribute to the ${\cal C}_Q(\theta)$. The 
${\cal C}_Q(\theta)$ is a Hermitian positive semidefinite operator with eigenprojectors the $\varpi _{\theta}(i|i+1,...,d^2)$ 
(the ones which are non-zero), and eigenvalues the corresponding $dQ(i\;|\;\theta)$.
There is a finite number of sets of projectors $\{\varpi _{\theta}(i|i+1,...,d^2)\}$ (given in Eqs(\ref{nm1}),(\ref{nm2}) below), 
and therefore the set of all ${\cal C}_Q(\theta)$ is a subset ${\cal C}_d$ of ${\cal M}_d$.
The Choquet integral is a map from ${\cal M}_d$ to ${\cal C}_d$.

In the case of 
coherent states with generic fiducial vectors, only the $d$ highest values of $Q(i\;|\;\theta)$ enter in Eq.(\ref{010}). 
In the rest of the paper we consider generic fiducial vectors, and
\begin{eqnarray}\label{010A}
{\cal C}_Q(\theta)&=&\sum_{i=d^2-d+1}^{d^2}dQ(i\;|\;\theta)\varpi _{\theta}(i|i+1,...,d^2)
\end{eqnarray}
We refer to the $\ket{C;\alpha _i, \beta _i}$, $\Pi(\alpha _i, \beta _i)$, $Q(i\;|\;{\theta })$
with $i=d^2-d+1,...,d$, which enter in Eq.(\ref{010A}),  as dominant coherent states, dominant projectors and dominant values of the $Q$-function, 
for the operator $\theta$.
We also refer to the $\ket{C;\alpha _i, \beta _i}$, $\Pi(\alpha _i, \beta _i)$, $Q(i\;|\;{\theta })$
in the `tail' $i=1,...,d^2-d+1,$ as inferior coherent states, inferior projectors and inferior values of the $Q$-function, for the operator $\theta$.

If two of the dominant values of the $Q$-function, $Q(i\;|\;\theta)$ and $Q(i+1\;|\;\theta)$,  are equal to each other, there are two different orderings of the corresponding coherent states that can be used in Eq.(\ref{010A}),
and they both lead to the same result.
Indeed, the contribution of these two terms to ${\cal C}_Q(\theta)$, is
\begin{eqnarray}\label{ppp}
&&dQ(i\;|\;\theta)\varpi _{\theta}(i|i+1,...,d^2)+dQ(i+1\;|\;\theta)\varpi _{\theta}(i+1|i+2,...,d^2)\nonumber\\&&=
dQ(i\;|\;\theta)[\Pi _{\theta}(i;{i+1};...;{d^2})-\Pi _{\theta}({i+2};...; {d^2})]
\end{eqnarray}
The $\Pi _{\theta}(i;{i+1};...;{d^2})$ does not change if we swap the two coherent states ranked with $i$ and $i+1$.
In this case there is a degeneracy in the eigenvalues of ${\cal C}_Q(\theta)$.
$Q(i\;|\;\theta)$ and $Q(i+1\;|\;\theta)$ are two eigenvalues equal to each other,
and $\Pi _{\theta}(i;{i+1};...;{d^2})-\Pi _{\theta}({i+2};...; {d^2})$ is the corresponding eigenprojector to a two dimensional space.

It is easily seen that ${\cal C}_Q(a \theta)=a{\cal C}_Q(\theta)$ for $a \ge 0$.
In general ${\cal C}_Q(\theta)+{\cal C}_Q(\phi)\ne {\cal C}_Q(\theta +\phi)$.
The question under what conditions we have additivity,
leads naturally to the concept of comonotonicity, which we discussed in a classsical context earlier, and which is discussed in a quantum context later.

\begin{proposition}\label{www}
${\cal C}_Q(\theta)$ can be written as
\begin{eqnarray}\label{n9b}
&&{\cal C}_Q(\theta)={\cal C}_{Q,1}(\theta)+{\cal C}_{Q,2}(\theta)+...+{\cal C}_{Q,{d^2}}(\theta)\nonumber\\
&&{\cal C}_{Q,1}(\theta)=d\sum \Pi(i)Q(i|\;\theta)\nonumber\\
&&{\cal C}_{Q,2}(\theta)=d\sum {\mathfrak D}(i,j)\min \{Q(i|\;\theta),Q(j|\;\theta)\}\nonumber\\
&&{\cal C}_{Q,3}(\theta)=d\sum {\mathfrak D}(i,j,k)\min \{Q(i|\;\theta),Q(j|\;\theta), Q(k|\;\theta)\}\nonumber\\
&&...............................................................\nonumber\\
&&{\cal C}_{Q,{d^2}}(\theta)=d{\mathfrak D}({1,...,{d^2}})Q(1|\;\theta)
\end{eqnarray}
\end{proposition}
\begin{proof}
We start from Eq.(\ref{010}), and we group together all terms that involve 
the operators ${\mathfrak D}$ with $k$ Hilbert spaces. This gives the ${\cal C}_{Q,k}(\theta)$. 
\end{proof}
Choquet integrals, are designed for cases where the various alternatives 
are not independent, but they overlap with each other. The M\"obius transforms studied in section \ref{MMM}, quantify these overlaps.
The term $d{\mathfrak D}(i,j)\min \{Q(i|\;\theta),Q(j|\;\theta)\}$ is a `correction' related to the overlap between two coherent states $i,j$.
The term $d{\mathfrak D}(i,j,k)\min \{Q(i|\;\theta),Q(j|\;\theta),Q(k|\;\theta) \}$ 
is a `correction' related to the overlap between three coherent states $i,j,k$, etc.
Adding all of them together, we remove the double-counting due to overlaps between the coherent states.

The following proposition gives the ${\cal C}_Q(\theta)$ in some special cases.
\begin{proposition}
\mbox{}
\begin{itemize}
\item[(1)]
If the $d$ dominant values of $Q(i\;|\;\theta)$ are equal to each other, then 
\begin{eqnarray}\label{AA}
{\cal C}_Q(\theta)=d\max\{Q(\alpha, \beta|\theta)\}{\bf 1}.
\end{eqnarray}
\item[(2)]
Let $\theta =\sum _m \lambda _m\ket{X;m}\bra{X;m}$ with $\lambda _m\ge 0$ (so that it is positive semidefinite operator). 
Then
\begin{eqnarray}\label{33}
{\cal C}_Q(\theta)=d\max\{Q(\alpha, \beta|\theta)\}{\bf 1}.
\end{eqnarray}
Similar result holds for $\theta=\sum _m \lambda _m\ket{P;m}\bra{P;m}$.
\item[(3)]
\begin{eqnarray}
{\cal C}_Q({\bf 1})={\bf 1};\;\;\;\;\;
{\cal C}_Q(\theta +\lambda {\bf 1})={\cal C}_Q(\theta)+\lambda{\bf 1}.
\end{eqnarray}
\end{itemize}
\end{proposition}
\begin{proof}
\mbox{}
\begin{itemize}
\item[(1)]
This follows immediately from Eq.(\ref{010A}) and the fact that the $\varpi _{\theta}(i|i+1,...,d^2)$ are an orthogonal and complete set of projectors.
\item[(2)]
For $\theta =\sum _m \lambda _m\ket{X;m}\bra{X;m}$, the $D(\alpha , \beta)\theta [D(\alpha , \beta)]^{\dagger}$ do not depend on $\alpha$, and 
consequently the $d$ dominant values of $Q(\alpha, \beta|\theta)$ are equal to each other.
From this follows Eq.(\ref{33}).
\item[(3)]
This follows from the fact that $Q(\alpha, \beta\;|\;{\bf 1})=\frac{1}{d}$.
\end{itemize}
\end{proof}

One of the applications of the Choquet integral is that it provides bounds for various physical quantities.
The following proposition provides bounds to ${\rm Tr}(\theta )$, 
${\rm Tr}(\rho \theta )$ (where $\rho$ is a density matrix), and ${\rm Tr}(\theta \phi)$,
in terms of ${\rm Tr}[{\cal C}_Q(\theta)]$, ${\rm Tr}[{\cal C}_Q(\phi)]$.
It also shows that ${\rm Tr}[{\cal C}_Q(\theta)]$ is a convex function.
We note that the calculation of ${\rm Tr}[{\cal C}_Q(\theta)]$ only requires the calculation of the $Q$-function and its ranking in Eq.(\ref{56}).
It does not require the calculation of the projectors $\varpi _{\theta}(i|i+1,...,d^2)$.
Indeed, from Eq.(\ref{010A}) it follows that
\begin{eqnarray}\label{rfv}
{\rm Tr}[{\cal C}_Q(\theta)]&=&\sum_{i=d^2-d+1}^{d^2}dQ(i\;|\;\theta).
\end{eqnarray}

\begin{proposition}\label{1234}
Let $\theta, \phi \in {\cal M}_d$.  
\begin{itemize}
\item[(1)]
For $\theta \ne 0$
\begin{eqnarray}\label{vvv}
\frac{1}{d}{\rm Tr}[{\cal C}_Q(\theta)]<{\rm Tr}(\theta)\le {\rm Tr}[{\cal C}_Q(\theta)].
\end{eqnarray}
For $\theta ={\bf 1}$ the right hand side inequality becomes equality.
The left hand side is a strict inequality.
\item[(2)]
For any density matrix $\rho$
\begin{eqnarray}
{\rm Tr}(\rho \theta)\le {\rm Tr}[{\cal C}_Q(\theta)].
\end{eqnarray}
\item[(3)]
\begin{eqnarray}\label{333}
{\rm Tr}(\theta \phi)\le{\rm Tr}[{\cal C}_Q(\theta)]{\rm Tr}[{\cal C}_Q(\phi)].
\end{eqnarray}
\item[(4)]
${\rm Tr}[{\cal C}_Q(\theta)]$ is a convex function:
\begin{eqnarray}\label{33b}
{\rm Tr}[{\cal C}_Q(a\theta +(1-a)\phi)]\le a{\rm Tr}[{\cal C}_Q(\theta)]+(1-a){\rm Tr}[{\cal C}_Q(\phi)];\;\;\;\;0\le a \le 1.
\end{eqnarray}
\end{itemize}
\end{proposition}
\begin{proof}
\mbox{}
\begin{itemize}
\item[(1)]
\begin{eqnarray}
{\rm Tr}[{\cal C}_Q(\theta)]-{\rm Tr}(\theta)=(d-1)\sum _{i=d^2-d+1}^d Q(i|\theta)- \sum _{i=1}^{d^2-d} Q(i|\theta)
\end{eqnarray}
There are $d^2-d$ terms in both of these sums and any term in the first sum is greater or equal to any term in the second sum.
This proves that ${\rm Tr}(\theta)\le {\rm Tr}[{\cal C}_Q(\theta)]$.  Also
\begin{eqnarray}\label{q1n}
d{\rm Tr}(\theta)-{\rm Tr}[{\cal C}_Q(\theta)]=d\sum _{i=1}^{d^2-d} Q(i|\theta)
\end{eqnarray}
For $\theta \ne 0$ this is always a positive number. Indeed, $\theta$ is a positive semidefinite operator and $\theta =\sum \theta _\nu {\mathfrak P}_\nu$ where $\theta _\nu\ge 0$ are its eigenvalues, and 
${\mathfrak P}_\nu$ its eigenprojectors. In this case 
\begin{eqnarray}
Q(i|\theta)=\sum _{\nu}\theta _\nu Q(i|{\mathfrak P}_\nu).
\end{eqnarray}
We have explained earlier that for generic fiducial vectors, 
$d$ or more coherent states form a total set of vectors in $H(d)$. Therefore for every $\nu$, there exists at least one $i$ for which the $Q(i|{\mathfrak P}_\nu)$ is positive, and then
the left hand side of Eq.(\ref{q1n}) is positive.
\item[(2)]
For Hermitian positive semidefinite operators $A,B$, it is known\cite{ineq} that ${\rm Tr}(AB)\le {\rm Tr}(A){\rm Tr}(B)$.
We use this in conjuction with Eq.(\ref{vvv}), and we get:
\begin{eqnarray}
{\rm Tr}(\rho \theta )\le {\rm Tr}(\rho){\rm Tr}(\theta )={\rm Tr}(\theta )\le{\rm Tr}[{\cal C}_Q(\theta)].
\end{eqnarray}
\item[(3)]
We use the formula ${\rm Tr}(AB)\le {\rm Tr}(A){\rm Tr}(B)$ in conjuction with Eq.(\ref{vvv}), and we get:
\begin{eqnarray}
{\rm Tr}(\theta \phi)\le {\rm Tr}(\theta ){\rm Tr}(\phi)\le {\rm Tr}[{\cal C}_Q(\theta)]{\rm Tr}[{\cal C}_Q(\phi)].
\end{eqnarray}
\item[(4)]
We first prove that
\begin{eqnarray}\label{33c}
{\rm Tr}[{\cal C}_Q(\theta +\phi)]\le {\rm Tr}[{\cal C}_Q(\theta)]+{\rm Tr}[{\cal C}_Q(\phi)].
\end{eqnarray}
We start from the relation
\begin{eqnarray}
Q(d^2|\theta +\phi)=Q(i_1|\theta)+Q(j_1|\phi).
\end{eqnarray}
Here the coherent state labelled with $d^2$ in the ordering of $Q(\alpha, \beta |\theta +\phi)$, is labelled with $i_1$ in the ordering of
$Q(\alpha, \beta |\theta)$, and with $j _1$ in the ordering of $Q(\alpha, \beta |\phi)$.
Similarly
\begin{eqnarray}
Q(d^2-1|\theta +\phi)=Q(i_2|\theta)+Q(j_2|\phi),
\end{eqnarray}
etc. Adding these equations we get
\begin{eqnarray}
{\rm Tr}[{\cal C}_Q(\theta +\phi)]=d\sum _{\ell =1}^d[Q(i_\ell|\theta)+Q(j_\ell|\phi)].
\end{eqnarray}
The indices $i_1,...,i_d$ are different from each other, and take values in the set $A\cup B$ where $A=\{1,...,d^2-d\}$ and $B=\{d^2-d+1,...,d^2\}$.
If the  $Q(i_\ell|\theta)$ has index $i_\ell \in A$, we replace it with another $Q(i'_\ell|\theta)$ with index $i'_\ell$ in $B-\{i_1,...,i_d\}$ 
(in a way that at the end all indices are diferent from each other). This increases the sum, and therefore 
\begin{eqnarray}
d\sum _{\ell =1}^dQ(i_\ell|\theta) \le {\rm Tr}[{\cal C}_Q(\theta)].
\end{eqnarray}
We do the same with the $Q(j_\ell|\phi)$ and we prove Eq.(\ref{33c}).
From this follows easily Eq.(\ref{33b}).
\end{itemize}
\end{proof}

We define the `dominance ratio'
\begin{eqnarray}\label{ratio}
r(\theta)=\frac{{\rm Tr}[{\cal C}_Q(\theta)]}{d{\rm Tr}(\theta)}=\frac{\sum \limits _{i=d^2-d+1}^{d^2} Q(i|\theta)}{\sum \limits _{i=1}^{d^2} Q(i|\theta)};\;\;\;\;
\frac{1}{d}\le r(\theta)<1.
\end{eqnarray}
It gives the percentage of the sum of the dominant values of the $Q$-function, with respect to the sum of all values of the $Q$-function. 
For any $\lambda >0$, the operators $\theta$ and $\lambda \theta$ have the  same dominance ratio.
In examples later, we present values of this quantity.
\begin{proposition}
Displacement transformations on $\theta$, imply displacement transformations on ${\cal C}_Q(\theta)$:
\begin{eqnarray}\label{hhjj}
{\cal C}_Q\left [D(\alpha , \beta)\theta D^{\dagger}(\alpha , \beta)\right ]=D(\alpha , \beta){\cal C}_Q(\theta ) D^{\dagger}(\alpha , \beta)
\end{eqnarray}
\end{proposition}
\begin{proof}
We first use the definition of Eq.(\ref{mmm}) in conjuction with Eq.(\ref{1111}) to prove that 
\begin{eqnarray}\label{n9}
Q\left [\gamma, \delta| D(\alpha , \beta)\theta D^{\dagger}(\alpha , \beta)\right ]=Q(\gamma-\alpha, \delta-\beta|\theta )
\end{eqnarray}
Then we use Eq.(\ref{010A}) in conjuction with Eq.(\ref{1019}) (which are generalized for all $\varpi _{\theta}(i|i+1,...,d^2)$),
and prove Eq.(\ref{hhjj}).
\end{proof}

The Choquet integral ${\cal C}_Q(\theta )$ is based on a ranking formalism and it depends strongly on the 
dominant coherent states that give a high value of the $Q$-function. 
Under displacement transformations, the ${\cal C}_Q(\theta )$ transforms as in Eq.(\ref{hhjj}).
In contrast to this, the Wehrl entropy does not change (Eq.(\ref{bbb})). The Wehrl entropy
shows whether the $Q$-function is uniform or concentrated in a few coherent states, but in the latter case it does not show where it is concentrated.
Furthermore, ${\rm Tr}[{\cal C}_Q(\theta)]$ is a convex function while entropy is a concave function, i.e.,
mixing of two density matrices $\theta$, $\phi$ into
$a\theta +(1-a)\phi$,  decreases the ${\rm Tr}({\cal C}_Q)$ and increases the entropy.
Therefore the Choquet integral contains complementary information to the Wehrl entropy.

\subsection{Robustness of the formalism in the presence of noise:}\label{noise}

The Choquet formalism is robust in the presence of noise.
This is because the formalism is based on the ranking in Eq.(\ref{56}).
Noise affects all values of the $Q$-function in approximately equal way, and it is unlikely that it will change the ranking drastically.

We present a numerical example which shows this.
In the $3$-dimensional space $H(3)$, we consider a Hermitian operator $\theta$, and add noise in its elements as follows:
\begin{eqnarray}\label{97}
\theta=\left(
\begin{array}{ccc}
8+r_1&1+r_2+ir_3&-5+r_4+ir_5\\
1+r_2-ir_3&4+r_6&2+r_7+ir_8\\
-5+r_4-ir_5&2+r_7-ir_8&7+r_9
\end{array}
\right )
\end{eqnarray}
$r_1,...,r_9$ are uniformly distributed random numbers in the region $(-1,1)$.
We have calculated the eigenvalues $e_1,e_2, e_3$ (where $e_1\le e_2 \le e_3$), the corresponding eigenvectors $\ket{v_1},\ket{v_2},\ket{v_3}$, and the function $Q(\alpha, \beta )$,
of this operator.
In table \ref{t1}, we present results for the case without noise (first row), and for five cases with noise.
The three dominant values of $Q(\alpha, \beta)$, the eigenvalues, and the dominance ratio $r(\theta)$  (Eq.(\ref{ratio})) are shown.
For the eigenvectors we show their overlaps $\tau _i=|\bra{u_i}v_i\rangle |^2$ with their counterparts $\ket{u_i}$ in the noiseless case.

It is seen that the dominant coherent states and the corresponding dominant values of $Q(\alpha, \beta|\theta)$ change only slightly.
The lowest eigenvalue  is sensitive to noise.
Overall, the Choquet formalism is robust in the presence of noise.

\section{Comonotonic operators}\label{L}

We generalize the concept of comonotonic functions discussed in section \ref{nnmm}, into operators in ${\cal M}_d$.
Comonotonic operators is one way of making precise the intuitive concept of physically similar operators.
This is analogous to students with comonotonic marks in section \ref{aatt}, which have similar academic strengths and weaknesses.
The Choquet integral of the sum of comonotonic operators, is equal to the sum of the Choquet integrals of the operators.
This is used in the next section (corollary \ref{c2}), to derive bounds for the trace of Choquet integrals, which 
physically are related to mild changes in the physical system.

\begin{definition}
Two operators $\theta ,\phi \in {\cal M}_d$ are comonotonic, if the following statements, which are equivalent to each other, hold:
\begin{itemize}
\item[(1)]
The ranking permutation of Eq.(\ref{56}) is the same for both operators: $\sigma(\alpha, \beta\;|\;\theta)=\sigma(\alpha, \beta\;|\;\phi)$. The
$\theta$, $\phi$ have the same dominant projectors, and the corresponding ${\cal C}_Q(\theta), {\cal C}_Q(\phi)$ have the same eigenprojectors and commute:
\begin{eqnarray}\label{eigen}
\varpi _{\theta}(i|i+1,...,d^2)=\varpi _{\phi}(i|i+1,...,d^2);\;\;\;\;\;\;\;[{\cal C}_Q(\theta),{\cal C}_Q(\phi)]=0.
\end{eqnarray}
\item[(2)]
for the $d$ dominant values of $Q(\alpha, \beta|\;\theta)$ and $Q(\alpha, \beta|\phi)$ 
\begin{eqnarray}\label{66}
[Q(\alpha, \beta|\;\theta)-Q(\gamma, \delta|\;\theta)][Q(\alpha, \beta|\;\phi)- Q(\gamma, \delta|\;\phi)]\ge 0
\end{eqnarray}
\end{itemize}
\end{definition}

It is easily seen that: 
\begin{itemize}
\item
The ${\bf 1}$ is comonotonic to any other operator. 
\item
For $\lambda \ge 0$, the $\theta$ and $\lambda \theta$ are comonotonic.
\item
If $\theta, \phi$ are comonotonic and $\lambda, \mu \ge 0$, then the $\theta, \phi, \lambda \theta+\mu \phi$ are pairwise comonotonic.
\item
If $\theta, \psi$ are comonotonic, and $\phi, \psi$ are comonotonic, then the  $\theta +\phi, \psi$ are comonotonic.
\item
If ${\cal C}_Q(\theta)={\cal C}_Q(\phi)$, then the operators $\theta, \phi$ are comonotonic.
\end{itemize}

\begin{proposition}
If $\theta ,\phi$ are comonotonic operators, then 
\begin{eqnarray}\label{com}
{\cal C}_Q(a\theta +b\phi)=a{\cal C}_Q(\theta)+b{\cal C}_Q(\phi);\;\;\;\;\;\;a,b\ge 0.
\end{eqnarray}
\end{proposition}
\begin{proof}
For comonotonic operators $\theta, \phi$, the eigenprojectors $\{\varpi _{\theta}(i|i+1,...,d^2)\}$ of ${\cal C}_Q(\theta)$ are the same with eigenprojectors 
$\{\varpi _{\phi}(i|i+1,...,d^2)\}$ of ${\cal C}_Q(\phi)$ (Eq.(\ref{eigen})). Then Eq.(\ref{com}) follows easily.
\end{proof}
Additivity holds only for comonotonic operators, and
we refer to this as the weak additivity property of Choquet integrals.

\begin{proposition}
If $\theta ,\phi$ are comonotonic operators, then the 
$D(\alpha , \beta)\theta D^{\dagger}(\alpha , \beta)$ and $D(\alpha , \beta)\phi D^{\dagger}(\alpha , \beta)$
are also comonotonic operators.
\end{proposition}
\begin{proof}
Since $\theta, \phi$ are $Q$-comonotonic
\begin{eqnarray}
[Q(\gamma,\delta\;|\;\theta)-Q(\epsilon, \zeta\;|\;\theta)][Q(\gamma,\delta\;|\;\phi)- Q(\epsilon, \zeta\;|\;\phi)]\ge 0.
\end{eqnarray}
We insert $\gamma= \gamma '-\alpha$, $\delta=\delta '-\beta$, $\epsilon=\epsilon '-\alpha$, $\zeta=\zeta '-\beta$ and we get
\begin{eqnarray}
[Q(\gamma '-\alpha,\delta '-\beta\;|\;\theta)-Q(\epsilon '-\alpha, \zeta '-\beta \;|\;\theta)]
[Q(\gamma '-\alpha,\delta '-\beta\;|\;\phi)- Q(\epsilon '-\alpha, \zeta '-\beta\;|\;\phi)]\ge 0.
\end{eqnarray}
Taking into account Eq.(\ref{n9}) we rewrite this as
\begin{eqnarray}
&&[Q(\gamma',\delta'\;|\;D(\alpha , \beta)\theta D^{\dagger}(\alpha , \beta))-Q(\epsilon', \zeta'\;|\;D(\alpha , \beta)\theta D^{\dagger}(\alpha , \beta))]
\nonumber\\&&\times
[Q(\gamma',\delta'\;|\;D(\alpha , \beta)\phi D^{\dagger}(\alpha , \beta))- Q(\epsilon', \zeta'\;|\;D(\alpha , \beta)\phi D^{\dagger}(\alpha , \beta))]\ge 0.
\end{eqnarray}
and this proves that the 
$D(\alpha , \beta)\theta D^{\dagger}(\alpha , \beta)$ and $D(\alpha , \beta)\phi D^{\dagger}(\alpha , \beta)$
are comonotonic operators.
\end{proof}

\subsection{Equivalence classes of comonotonic operators in ${\cal M}_d'$}

Comonotonicity is not transitive in the set ${\cal M}_d$. 
For example, ${\bf 1}$ is comonotonic to every operator and yet there are operators which are not comonotonic.
We define a subset of ${\cal M}_d$ where comonotonicity is transitive.
\begin{definition}
\mbox{}
\begin{itemize}
\item[(1)]
${\cal M}_d'$ is a subset of ${\cal M}_d$ which contains operators $\theta$ for which the $d$ dominant values of $Q(i\;|\;\theta)$ are different from each other.
${\cal N}_d'\subset {\cal M}_d'$ is the set of such operators with trace equal to one.
\item[(2)]
Through the Choquet integral map, ${\cal M}_d'$ is mapped into ${\cal C}_d' \subset {\cal C}_d$ which contains Choquet integrals with eigenvalues $Q(i\;|\;\theta)$ which are different from each other.
 
\end{itemize}
\end{definition}
Comonotonicity is transitive in ${\cal M}_d'$. In this case  we have a strict inequality in Eq.(\ref{66}).
This is analogous to commutativity which is not transitive in general, but it is transitive if we restrict ourselves to matrices with eigenvalues which are different from each other. 

In ${\cal M}_d'$ (and ${\cal N}'_d$) comonotonicity is an equivalence relation, which partitions ${\cal M}_d'$ (and ${\cal N}'_d$) into equivalence classes, which we denote as 
${\cal M}_d'(\nu)$ (and ${\cal N}_d'(\nu)$). We denote comonotonic operators in these classes with $\theta _1 \sim \theta _2$.
It is easily seen that if $\theta _1 \sim \theta _2$ then $a\theta _1+b\theta _2 \sim \theta _1 \sim \theta _2$, where $a,b\ge 0$.

Through the Choquet integral map, ${\cal C}_d'$ is also partitioned into equivalence classes, which we denote as 
${\cal C}_d'(\nu)$.
There is an ordered set of $d$ coherent states associated with each equivalence class ${\cal M}_d'(\nu)$.
The number of such classes is
\begin{eqnarray}\label{nm1}
n_d=d^2(d^2-1)...(d^2-d+1)=\frac{d^2!}{(d^2-d)!}
\end{eqnarray}
We prove this by taking one coherent state from the set of $d^2$ coherent states, and then another coherent state from the remaining set of $d^2-1$ coherent states (which we use together with the first coherent state for the two-dimensional cumulative projectors in section \ref{eac}), etc.

So there is a finite number of sets of projectors in the formalism:
\begin{eqnarray}\label{nm2}
{\cal S}_{\nu}=\{\varpi _{\nu}(i|i+1,...,d^2)]\;|i=d^2-d+1,...,d^2\};\;\;\;\;\;\nu=1,...,n_d.
\end{eqnarray}
All Choquet integrals in the same equivalence class ${\cal C}_d'(\nu)$ commute with each other, and have the same eigenprojectors:
Through the Choquet integral map, the property comonotonicity in ${\cal M}_d'$, becomes commutativity in ${\cal C}_d'$. 

In the case that two of the dominant values of the $Q$-function are equal to each other (i.e., for operators $\theta \in {\cal M}_d-{\cal M}_d'$),
the corresponding sums of projectors 
\begin{eqnarray}
\varpi _{\nu}(i|i+1,...,d^2)+\varpi _{\nu}(i+1|i+2,...,d^2)=\Pi _{\theta}(i;{i+1};...;{d^2})-\Pi _{\theta}({i+2};...; {d^2})
\end{eqnarray}
enter into ${\cal C}_Q(\theta)$ as explained in Eq.(\ref{ppp}).

\subsection{Comonotonicity intervals of operators $\theta (\lambda)$ and crossings of the $Q$-function}\label{NN}

In many cases the operator $\theta $ is a function of a real parameter $\lambda$.
Examples are:
\begin{itemize}
\item
A Hamiltonian $\theta (\lambda)=\theta _1+\lambda \theta _2$, where $\theta _1$ is the free part, $\theta _2$ the interaction part, and $\lambda$ the coupling constant
(an example is given in section \ref{example} below).
\item
$\theta (\lambda)=\ket{g(\lambda)}\bra{g(\lambda)}$ where $\ket{g(\lambda)}$ is the ground state of a system described by a Hamiltonian
${\mathfrak H}(\lambda)$ with coupling constant $\lambda$ (two examples are given in section \ref{ex}, with and without degeneracies in the eigenvalues).
\item
The $\theta (\lambda)=\exp (-\lambda {\mathfrak H})$ where ${\mathfrak H}$ is a Hamiltonian and $\lambda$ the inverse temperature. The trace of this operator is the partition function, and bounds for it are given
in section \ref{ex} below.
\end{itemize}
If $\theta (\lambda)$ is a continuous function of $\lambda$, the $Q[\alpha, \beta|{\theta}(\lambda)]$ are also continuous functions of $\lambda$.
Consequently, there are intervals of the parameter $\lambda$, where the ranking of the 
$d$ highest values of $Q[\alpha, \beta |\theta (\lambda)]$ remains unchanged. 
We call them comonotonicity intervals. By definition, if $\lambda _1,\lambda _2$ belong to the same comonotonicity interval, 
then the $\theta (\lambda _1)$, $\theta (\lambda _2)$ are comonotonic.
But  it is not necessary that all operators in a comonotonicity interval belong in the same equivalence class of ${\cal M}_d'$
(there are pairs of comonotonic operators in ${\cal M}_d-{\cal M}_d'$, but transitivity might not hold).

There might be values $\lambda _i$ where we have crossings of the $d$ highest values of the $Q$-function:
\begin{eqnarray}\label{cro}
Q[\alpha _1, \beta _1|\theta (\lambda _i)]=Q[\alpha _2, \beta _2|\theta (\lambda _i)]
\end{eqnarray}
We call them crossings of the $Q$-function. At these points a change of the ranking occurs,
and the matrix ${\cal C}_Q[\theta (\lambda)]$ has a discontinuity.
The ${\rm Tr}\{{\cal C}_Q[\theta (\lambda)]\}$ is continuous at these points (as sum of continuous functions),
but its derivative with respect to $\lambda$, might have discontinuities.

If $\{\lambda _i\}$ is the set of the crossings of the $Q$-function, 
the $\lambda$-axis is partitioned to many intervals $(\lambda _i,\lambda _{i+1})$, and within each interval
all the $\theta(\lambda)$ are pairwise comonotonic operators.

\begin{remark}
A phenomenon analogous to `avoided crossing' of the energy levels, might occur. 
A small external perturbation can invalidate the equality in Eq.(\ref{cro}).
For example, a small amount of noise added into $\theta (\lambda)$ will make it $\theta (\lambda)+\Delta \theta$
where $\Delta \theta$ is an infinitesimal matrix (which we assume to be Hermitian).
Then 
\begin{eqnarray}
Q[\alpha _i, \beta _i|\theta (\lambda)+\Delta \theta]=Q[\alpha _i, \beta _i|\theta (\lambda)]+Q[\alpha _i, \beta _i|\Delta \theta]
\end{eqnarray}
and in general $Q[\alpha _1, \beta _1|\Delta \theta]\ne Q[\alpha _2, \beta _2|\Delta \theta]$.
In this case the curve $Q[\alpha _1, \beta _1|\theta (\lambda)+\Delta \theta]$ on the left of the crossing point $\lambda _i$ 
will join the curve $Q[\alpha _2, \beta _2|\theta (\lambda)+\Delta \theta]$ on the right of the crossing point
(and the curve $Q[\alpha _2, \beta _2|\theta (\lambda)+\Delta \theta]$ on the left of the crossing point
will join the curve $Q[\alpha _1, \beta _1|\theta (\lambda)+\Delta \theta]$ on the right of the crossing point).
Therefore $\lambda _i$ will not be a crossing point of the $Q$-function, and the left and right comonotonicity intervals, will join
to become one comonotonicity interval.
We call this `avoided crossings of the $Q$-function'.
It is a phenomenon which should be studied in its own right.
In the examples below, we assume the absence of such perturbations, and the absence of the `avoided crossings of the $Q$-function'. 
\end{remark}

\subsection{Example}\label{example}

We consider the operator (in the position basis):
\begin{eqnarray}\label{1000}
\theta (\lambda)=\theta _1+\lambda \theta _2;\;\;\;\;\;
\theta_1=\left(
\begin{array}{ccc}
6&0&i\\
0&12&0\\
-i&0&15
\end{array}
\right );\;\;\;\;\;\;
\theta_2=\left(
\begin{array}{ccc}
7&3&6\\
3&7&0\\
6&0&7
\end{array}
\right ).
\end{eqnarray}
The values of $Q(\alpha, \beta\;|\theta)$ are:
\begin{eqnarray}
&&Q(0,0|\theta )=4.5 +3.476\lambda ;\;\;\;\;\;Q(0,1|\theta )=3+4.476\lambda;\;\;\;\;\;\;Q(0,2|\theta )=3.5+3.762 \lambda\nonumber\\
&&Q(1,0|\theta )=4.623+1.761\lambda;\;\;\;\;\;Q(1,1|\theta )=3.247+1.261\lambda;\;\;\;\;\;\;Q(1,2|\theta )=3.582+1.619\lambda \nonumber\\
&&Q(2,0|\theta )=4.376+1.762\lambda;\;\;\;\;\;Q(2,1|\theta )=2.752+1.261\lambda;\;\;\;\;\;\;Q(2,2|\theta )=3.417+1.619\lambda
\end{eqnarray}
The three dominant values, which in our notation are $Q(9|\theta), Q(8|\theta), Q(7|\theta)$, depend on the value of $\lambda$. We consider the interval $0\le \lambda \le 0.7$, and ordering of the
$Q(\alpha, \beta\;|\theta)$ shows that it consists of five comonotonicity intervals. 
In table \ref{t2} we show these comonotonicity intervals, the corresponding three dominant values of $Q(\alpha, \beta)$, and the 
dominance ratio $r[\theta (\lambda)]$.
The $r[\theta (\lambda)]$ is a continuous function of $\lambda$, but its derivative with respect to $\lambda$ has discontinuiuties at $\lambda=0.44$ and $\lambda=0.6$.

We present the Choquet integral, for the first two intervals.
In the comonotonicity interval $(0,0.06)$:
\begin{eqnarray}\label{b23}
&&\varpi _{\theta}(9)=\Pi(1,0);\;\;\;\;
\varpi _{\theta}(8)=\Pi(0,0; 1,0)-\Pi(1,0);\;\;\;\;
\varpi _{\theta}(7)={\bf 1}-\Pi(0,0;1,0)
\end{eqnarray}
Therefore 
\begin{eqnarray}
&&{\cal C}_Q[\theta(\lambda)]=3[Q(7|\theta)\varpi _{\theta}(7)+Q(8|\theta)\varpi _{\theta}(8)+Q(9|\theta)\varpi _{\theta}(9)]=A_1+\lambda B_1\nonumber\\
&&A_1=\left (
\begin{array}{ccc}
13.25&0.04-0.16i&0.01+0.14i\\
0.04+0.16i&13.53&-0.05-0.17i\\
0.01-0.14i&-0.05+0.17i&13.70\\
\end{array}
\right )\nonumber\\
&&B_1=\left (
\begin{array}{ccc}
6.28&1.31-i&1.18+0.13i\\
1.31+i&8.01&1.41+1.36i\\
1.18-0.13i&1.41-1.36i&6.7\\
\end{array}
\right );\;\;\;\;\;[A_1,B_1]=0..
\end{eqnarray}

In the comonotonicity interval $(0.06, 0.44)$, 
\begin{eqnarray}\label{b23}
\varpi _{\theta}(9)=\Pi(0,0);\;\;\;\;
\varpi _{\theta}(8)=\Pi(1,0; 0,0)-\Pi(0,0);\;\;\;\;\;
\varpi _{\theta}(7)={\bf 1}-\Pi(1,0;0,0)
\end{eqnarray}
Therefore 
\begin{eqnarray}
&&{\cal C}_Q[\theta(\lambda)]=3[Q(7|\theta)\varpi _{\theta}(7)+Q(8|\theta)\varpi _{\theta}(8)+Q(9|\theta)\varpi _{\theta}(9)]=A_2+\lambda B_2\nonumber\\
&&A_2=\left (
\begin{array}{ccc}
13.29&0.08-0.23i&0.01+0.15i\\
0.08+0.23i&13.62&-0.11-0.07i\\
0.01-0.15i&-0.11+0.07i&13.57
\end{array}
\right )\nonumber\\
&&B_2=\left (
\begin{array}{ccc}
5.65&0.73&1.10\\
0.73&6.75&2.20\\
1.10&2.20&8.59
\end{array}
\right );\;\;\;\;\;[A_2,B_2]=0.
\end{eqnarray}

Within each of the comonotonic intervals the $\theta(\lambda)$ are comonotonic operators, and
the ${\cal C}_Q[\theta(\lambda)]$ commute with each other and have the same eigenprojectors.
At the crossing points of the $Q$-function, the ${\cal C}_Q[\theta (\lambda)]$ has a discontinuity.

\section{Bounds for comonotonic operators}

We have seen earlier (proposition \ref{1234}) that the
trace of the Choquet integral is a bound for physical quantities like
${\rm Tr}(\theta)$, ${\rm Tr}(\rho \theta)$, ${\rm Tr}(\phi \theta)$, etc.
This is our physical motivation for using it in this section, to define an order among the Hermitian operators.

An order is useful if it has certain properties, and a natural property is that 
addition should preserve the order
(the analogue of $a\ge b$ implies that $a+c\ge b+c$ in real numbers).
We show that this property is valid in the case of comonotonic operators.
This already shows that in some sense comonotonic operators are physically similar.
More importantly, a whole family of comonotonic operators are bounded, with respect to this order (corollary \ref{c2} below).
This means that the trace of the Choquet integral is bounded, and therefore the other physical quantities to which this is a bound, are also bounded.

We introduce the `greater trace of the Choquet integral' preorder, as follows:
\begin{definition}
$\theta _1 \succ \theta _2$ if ${\rm Tr}[{\cal C}_Q(\theta _1)]\ge {\rm Tr}[{\cal C}_Q(\theta _2)]$.
\end{definition}
$ \succ$ is transitive, but the antisymmetry property does not hold ($\theta _1 \succ \theta _2$ and $\theta _1 \prec \theta _2$ 
implies that ${\rm Tr}[{\cal C}_Q(\theta _1)]={\rm Tr}[{\cal C}_Q(\theta _2)]$, but it does not follow that $\theta _1 = \theta _2$ ).
Therefore $\succ$ is a preorder, rather than a partial order.
It is a total preorder because for any $\theta _1, \theta _2$, either $\theta _1 \prec  \theta _2$ or $\theta _1\succ \theta _2$.

The following proposition shows that for comonotonic operators, addition preserves the $\succ$ preorder:
\begin{proposition}\label{c1}
\mbox{}
\begin{itemize}
\item[(1)]
If $\theta _1, \theta _3$ are comonotonic, and $\theta _2, \theta _3$ are comonotonic, then
\begin{eqnarray}\label{nc1}
\theta _1\succ \theta _2\;\;\rightarrow\;\; \theta _1+\theta _3\succ  \theta _2 +\theta _3.
\end{eqnarray}
\item[(2)]
For comonotonic $\theta _1, \theta _2 $
\begin{eqnarray}\label{nc2}
\theta _1\succ \theta _2\;\;\rightarrow\;\; \theta _1\succ  a\theta _1 +(1-a)\theta _2\succ \theta_2
;\;\;\;\;\;0\le a \le 1.
\end{eqnarray}
\end{itemize}
\end{proposition}
\begin{proof}
\mbox{}
\begin{itemize}
\item[(1)]
We have
\begin{eqnarray}
\theta _1\succ \theta _2\;\;\rightarrow\;\;{\rm Tr}[{\cal C}_Q(\theta _1)]\ge {\rm Tr}[{\cal C}_Q(\theta _2)]
\rightarrow\;\;{\rm Tr}[{\cal C}_Q(\theta _1)+{\cal C}_Q(\theta _3)]\ge {\rm Tr}[{\cal C}_Q(\theta _2)+{\cal C}_Q(\theta _3)]
\end{eqnarray}
Using the additivity of the Choquet integral for comonotonic operators, we rewrite this as
\begin{eqnarray}
{\rm Tr}[{\cal C}_Q(\theta _1+\theta _3)]\ge {\rm Tr}[{\cal C}_Q(\theta _2+\theta _3)],
\end{eqnarray}
and this proves the proposition.
\item[(2)]
$\theta _1\succ \theta _2$, implies that $(1-a)\theta _1\succ (1-a)\theta _2$.
We add $a\theta _1$ on both sides, and using Eq.(\ref{nc1}), and we get $\theta _1\succ  a\theta _1 +(1-a)\theta _2$.
In a similar way we prove that $a\theta _1 +(1-a)\theta _2\succ \theta_2$.
\end{itemize}
\end{proof}
\begin{cor}\label{c2}
Let $\theta (\lambda)$ be an operator which is a linear function of $\lambda$, within a comonotonicity interval $I$.
For $\lambda _1 < \lambda _2$ where $\lambda _1, \lambda _2 \in I$, we assume that $\theta (\lambda _1) \succ \theta (\lambda _2)$.
Then at any point $\lambda \in [\lambda _1, \lambda _2]$
\begin{eqnarray}\label{1qa}
\theta (\lambda _1)\succ  \theta (\lambda)\succ \theta (\lambda_2)
\end{eqnarray}
\end{cor}
\begin{proof}
For $\lambda \in [\lambda _1, \lambda _2]$
\begin{eqnarray}
\theta (\lambda)=a\theta (\lambda _1)+(1-a)\theta (\lambda _2);\;\;\;\;a=\frac{\lambda _2-\lambda }{\lambda _2-\lambda _1},
\end{eqnarray} 
and use of Eq.(\ref{nc2}) proves the statement.
\end{proof}
This result shows that for comonotonic operators, the 
${\rm Tr}\{{\cal C}_Q[\theta (\lambda)]\}$ is bounded by ${\rm Tr}\{{\cal C}_Q[\theta (\lambda _1)]\}$
and ${\rm Tr}\{{\cal C}_Q[\theta (\lambda _2)]\}$.
Therefore physical quantities to which  the ${\rm Tr}\{{\cal C}_Q[\theta (\lambda)]\}$ is a bound, are also bounded.

The corollary assumes that $\theta (\lambda)$ is a linear function of $\lambda$.
If $\theta (\lambda)$ is a non-linear function of $\lambda$, a comonotonicity interval can be divided into many small subintervals, and within each of them
$\theta (\lambda)$ is approximately a linear function of $\lambda$, and the corollary can be used.
An example of this is discussed in section \ref{ex12} below.

\begin{remark}
We rewrite the result in Eq.(\ref{1qa}), as
\begin{eqnarray}\label{1qaa}
f(\lambda _1)\ge f(\lambda)\ge f(\lambda _2);\;\;\;\;f(\lambda)={\rm Tr}\{{\cal C}_Q[\theta (\lambda)]\}.
\end{eqnarray}
Since $\theta (\lambda)$ is a continuous function of $\lambda$, the $f(\lambda)$ is also a continuous function of $\lambda$.
The intermediate value theorem for continuous functions states that 
for any $f_0\in [f(\lambda _2), f(\lambda _1)]$, there exists $\lambda _0\in [\lambda _1, \lambda _2]$ such that $f(\lambda _0)=f_0$. 
We point out that this is weaker than our result in Eq.(\ref{1qa}), for comonotonic operators.
\end{remark}

\section{Applications to the study of the ground state of physical systems}\label{ex}

The study of the ground state of a large physical system as a function of the coupling constant, is important for phase transitions.
We study a toy model which shows how our formalism can be used for the study of the ground state of a physical system.
For practical reasons, we consider a small system described with the $3$-dimensional space $H(3)$, and study two cases of Hamiltonians with and without
degeneracies in their eigenvalues.

\subsection{The ground state of a physical system with  Hamiltonian without degeneracies}\label{ex12}
We consider the following Hamiltonian which
is a non-linear function of the coupling constant $\lambda$:
\begin{eqnarray}\label{ham}
{\mathfrak H}(\lambda)=\left(
\begin{array}{ccc}
7&3i\lambda +\lambda ^2&6i\lambda +2\lambda ^2\\
-3i\lambda +\lambda ^2&9&5\lambda +4\lambda ^2\\
-6i\lambda +2\lambda ^2&5\lambda +4\lambda ^2&11
\end{array}
\right )
\end{eqnarray}
We have calculated numerically the eigenstate $\ket{g(\lambda)}$ which corresponds to the lowest eigenvalue of ${\mathfrak H}(\lambda)$, and then calculated the $Q[\alpha, \beta |{\mathfrak P}(\lambda)]]$ where
${\mathfrak P}(\lambda)=\ket{g(\lambda)}\bra{g(\lambda)}$. 
The three dominant values of $Q[\alpha, \beta |{\mathfrak P}(\lambda)]$  are shown in table \ref{t3} for $\lambda=0.1,0.2,..., 1$.
It is seen that the comonotonicity intervals are $[0, \lambda _1]$, $[\lambda _1,\lambda _2]$,  
$[\lambda _2,\lambda _3]$, and $[\lambda _3,1]$, where $\lambda _1\approx 0.3$, $\lambda _2\approx 0.4$ and $\lambda _3\approx 0.7$.
The dominance ratio $r[{\mathfrak P}(\lambda)]$ (Eq.(\ref{ratio})),  is also shown.

The ${\mathfrak H}(\lambda), {\mathfrak P}(\lambda)$ are non-linear functions of $\lambda$.
However if we divide each comonotonicity interval into small subintervals, we can assume that ${\mathfrak P}(\lambda)$ is approximately linear within each subinterval,
and use corollary \ref{c2}. 
For example we consider the comonotonicity interval $(0.5,0.7)$ where ${\mathfrak P}(0.5)\prec {\mathfrak P}(0.7)$ because ${\rm Tr}\{{\cal C}_Q[{\mathfrak P} (0.5)]\}=r[{\mathfrak P}(0.5)]=0.638$
and ${\rm Tr}\{{\cal C}_Q[{\mathfrak P} (0.7)]\}=r[{\mathfrak P}(0.7)]=0.710$ (table \ref{t3}).
We divide it into the subintervals $(0.5,0.6)$ and $(0.6,0.7)$, and then 
\begin{eqnarray}
&&0.5 \le \lambda \le 0.6\;\;\rightarrow\;\;{\mathfrak P}(0.5)\prec {\mathfrak P}(\lambda)\prec{\mathfrak P}(0.6)\nonumber\\
&&0.6\le \lambda \le 0.7\;\;\rightarrow\;\;{\mathfrak P}(0.6)\prec {\mathfrak P}(\lambda)\prec{\mathfrak P}(0.7)
\end{eqnarray}
Similarly in the comonotonicity interval $(0.8,1)$ we have ${\mathfrak P}(0.8)\succ {\mathfrak P}(1)$.
We divide it into the subintervals $(0.8,0.9)$ and $(0.9,1)$, and then 
\begin{eqnarray}
&&0.8\le \lambda \le 0.9\;\;\rightarrow\;\;{\mathfrak P}(0.8)\succ {\mathfrak P}(\lambda)\succ {\mathfrak P}(0.9)\nonumber\\
&&0.9\le \lambda \le 1\;\;\rightarrow\;\;{\mathfrak P}(0.9)\succ {\mathfrak P}(\lambda)\succ {\mathfrak P}(1).
\end{eqnarray}

The Wehrl entropy $E[{\mathfrak P}(\lambda)]$ (Eq.(\ref{entro})) that involves all $9$ values of $Q[\alpha, \beta |{\mathfrak P}(\lambda)]$, is also shown (we used natural logarithms and the result is in nats).
The Wehrl entropy has been used in the literature  as an indicator of phase transitions (e.g\cite{WE}).
The $E[{\mathfrak P}(\lambda)]$ has local maxima and minima at the values 
$\lambda _1'\approx 0.3$ and $\lambda _2'\approx 0.6$, which agrees roughly with the values where ${\cal C}_Q[{\mathfrak P}(\lambda)]$ is discontinuous.

The overlap $|\bra {g(0)}g(\lambda)\rangle|^2$ of the ground state $\ket{g(\lambda)}$ when the coupling constant is equal to $\lambda$, with
the ground state $\ket{g(0)}$ when the coupling constant is equal to $0$, is a measure of how much the ground state changes.
It is given in table \ref{t3}, and it is seen that the biggest change occurs in the region of $\lambda \sim (0.4, 0.7)$

Therefore different quantities confirm that at the crossing points of the $Q$-function, stronger physical changes occur into the system.

\subsection{The ground state of a physical system with Hamiltonian with degeneracies}\label{DD}
We consider the Hamiltonian:
\begin{eqnarray}\label{hamilt}
{\mathfrak H}(\lambda)=\left(
\begin{array}{ccc}
1.500&1.414+\lambda&1.732\\
1.414+\lambda^*&2.500&2.449\\
1.732&2.449&3.500
\end{array}
\right )
\end{eqnarray}
The eigenvalues $e_1,e_2,e_3$ of this Hamiltonian for $\lambda=-0.01, 0, 0.01, 0.01i, -0.01i$
are given in table \ref{t4}. For $\lambda=0$ we have a degeneracy, and the two lowest eigenvalues are equal to each other.

In the cases $\lambda=-0.01, 0.01, -0.01i, 0.01i$ that there is no degeneracy, we have calculated numerically the eigenstate $\ket{g(\lambda)}$ which corresponds to the lowest eigenvalue of ${\mathfrak H}(\lambda)$, and then calculated the $Q[\alpha, \beta |{\mathfrak P}(\lambda)]]$ where
${\mathfrak P}(\lambda)=\ket{g(\lambda)}\bra{g(\lambda)}$. We also calculated the Wehrl entropy $E[{\mathfrak P}(\lambda)]$.
For $\lambda =0$, the ${\mathfrak P}(0)$ is the projector to the two-dimensional eigenspace corresponding to the two lowest eigenvalues. 
In this case, we calculated the $Q[\alpha, \beta |\frac{1}{2}{\mathfrak P}(\lambda)]$ (so that the sum of all the $Q$-values is $1$), and it is these values that we used to calculate the Wehrl entropy.

The three dominant values of $Q[\alpha, \beta |{\mathfrak P}(\lambda)]$  are shown in table \ref{t4}.
Although $\lambda$ changes by a small amount, and the eigenvalues also change by a small amount, the 
dominant coherent states change drastically as we go from $\lambda=-0.01$ to $\lambda=0$ (where we get degeneracy), and then to $\lambda=0.01$
Similar comment can be made for going from $\lambda=-0.01i$ to $\lambda=0$, and then to $\lambda=0.01i$. 
If we compare the cases $\lambda=0.01$ and $\lambda=0.01i$, there is also a change in the 
dominant coherent states.
This is because the eigenvector corresponding to the lowest eigenvalue in the $\lambda=0.01$ case,
is very different from the eigenvector corresponding to the lowest eigenvalue in the $\lambda=0.01i$ case (we have
found that $|\langle{g(0.01i)}\ket{g(0.01)}|^2=0.5$). 

The results show that the method is sensitive enough to detect changes in the ground state in the case of degeneracies.
A change in the ground state, changes some values of the $Q$-function more than others, and this changes the ranking of the $Q$-function, and
for this reason it is easily detected by the Choquet formalism. In contrast to this, we have seen in section \ref{noise}, that
random noise affects all values of the $Q$-function in approximately equal way, the ranking remains the same,  
and for this reason the formalism is robust in the presence of noise.

\section{Bounds for partition functions}\label{partition}

Inequalities between quantities that involve matrices (e.g.,\cite{MA}) have many applications in Physics (e.g., \cite{IN1,IN2}), 
and also in other subjects like Control Theory in Electrical Engineering, Operational Research, etc.
In this general context, this paper uses Choquet integrals in conjuction with total sets of vectors (like coherent states).
In this section we derive bounds for the partition function, which together with proposition \ref{1234}, show the use of the formalism for bounds of physical quantities.

If $\theta$ is a Hamiltonian and $\lambda $ the inverse temperature, then ${\rm Tr}\exp (-\lambda \theta)$ is a partition function. 
Below we derive upper and lower bounds for the partition function.
We also show that ${\cal C}_Q[\exp (-\lambda \theta)]\ge \exp[-\lambda {\cal C}_Q(\theta)]$ (the $\theta _1 \ge \theta _2$ denotes the fact that
$\theta _1 - \theta _2$ is a positive semidefinite Hermitian operator).
\begin{proposition}
\mbox{}
\begin{itemize}
\item[(1)]
\begin{eqnarray}\label{q1}
&&{\rm Tr}{\cal C}_Q[\exp (-\lambda \theta)]\ge {\rm Tr}\exp (-\lambda \theta)\ge A\nonumber\\
&&A=\max\left (\frac{1}{d}{\rm Tr}{\cal C}_Q[\exp (-\lambda \theta)],\frac{1}{d}\sum _{\alpha,\beta}\exp[-d\lambda Q(\alpha,\beta|\theta)]\right).
\end{eqnarray}
\item[(2)]
If $\theta$ and $\exp (-\lambda \theta)$ (where $\lambda \ge 0$) are comonotonic operators, then
\begin{eqnarray}\label{MMM}
{\cal C}_Q[\exp (-\lambda \theta)]\ge \exp[-\lambda {\cal C}_Q(\theta)].
\end{eqnarray}
\end{itemize}
\end{proposition}
\begin{proof}
\mbox{}
\begin{itemize}
\item[(1)]
The left part of the inequality folows immediately from the inequality in Eq.(\ref{vvv}).
For the right part of the inequality, we have
\begin{eqnarray}\label{q2}
{\rm Tr}\exp (-\lambda \theta)=\sum _{\alpha,\beta}Q[\alpha,\beta|\exp (-\lambda \theta)].
\end{eqnarray}
The Bogoliubov inequality states that for any state $\psi$ and Hermitian operator $\phi$
\begin{eqnarray}
\bra{\psi}\exp (\phi)\ket{\psi}\ge \exp[\bra{\psi}\phi\ket{\psi}].
\end{eqnarray}
Consequently 
\begin{eqnarray}\label{4t5}
dQ[\alpha,\beta|\exp (-\lambda \theta)]\ge \exp[-d\lambda Q(\alpha,\beta|\theta)].
\end{eqnarray}
Therefore 
\begin{eqnarray}\label{q3}
\sum _{\alpha,\beta}Q[\alpha,\beta|\exp (-\lambda \theta)]\ge \frac{1}{d}\sum _{\alpha,\beta}\exp[-d\lambda Q(\alpha,\beta|\theta)].
\end{eqnarray}
From Eqs(\ref{q2}),(\ref{q3}), it follows that \cite{LIEB1}
\begin{eqnarray}
{\rm Tr}\exp (-\lambda \theta)\ge \frac{1}{d}\sum _{\alpha,\beta}\exp[-d\lambda Q(\alpha,\beta|\theta)].
\end{eqnarray}
But we also have
\begin{eqnarray}
{\rm Tr}\exp (-\lambda \theta)> \frac{1}{d}{\rm Tr}{\cal C}_Q[\exp (-\lambda \theta)],
\end{eqnarray}
from Eq.(\ref{vvv}). This completes the proof.
\item[(2)]
\begin{eqnarray}
{\cal C}_Q[\exp (-\lambda \theta)]=\sum _{i=d^2-d+1}^{d^2}dQ[\exp (-\lambda \theta)]\varpi _{\exp (-\lambda \theta)}(i|i+1,...,d^2).
\end{eqnarray}
Since the operators $\theta$ and $\exp (-\lambda \theta)$ are comonotonic
\begin{eqnarray}
\varpi _{\exp (-\lambda \theta)}(i|i+1,...,d^2)=\varpi _{\theta}(i|i+1,...,d^2).
\end{eqnarray}
Using Eq.(\ref{4t5}) which is based on the Bogoliubov inequality, we get
\begin{eqnarray}
&&\sum _{i=d^2-d+1}^{d^2}dQ[\exp (-\lambda \theta)]\varpi _{\exp (-\lambda \theta)}(i|i+1,...,d^2)\ge
\sum _{i=d^2-d+1}^{d^2}\exp[-d\lambda Q(\theta)]\varpi _{\theta}(i|i+1,...,d^2)\nonumber\\&&=
\exp \left \{\sum _{i=d^2-d+1}^{d^2}[-d\lambda Q(\theta)\varpi _{\theta}(i|i+1,...,d^2)]\right \}=\exp[-\lambda {\cal C}_Q(\theta)]
\end{eqnarray}
This completes the proof.

\end{itemize}
\end{proof}
There are two lower bounds in Eq.(\ref{q1}), which involve the $Q$-function of $\exp (-\lambda \theta)$ and the $Q$-function of $\theta$ .
We give two examples which show that sometimes the first is better lower bound, while other times the second is better lower bound.
The first example is
\begin{eqnarray}
\theta=\left(
\begin{array}{ccc}
8&1+i&-5\\
1-i&4&2\\
-5&2&7
\end{array}
\right )
\end{eqnarray}
For $\lambda =1$, we get
\begin{eqnarray}
{\rm Tr}\exp (-\lambda \theta)=0.440;\;\;\;\;\;
\frac{1}{d}{\rm Tr}{\cal C}_Q[\exp (-\lambda \theta)]=0.073;\;\;\;\;\;
\frac{1}{d}\sum _{\alpha,\beta}\exp[-d\lambda Q(\alpha,\beta|\theta)]=0.013.
\end{eqnarray}
Here the $\frac{1}{d}{\rm Tr}{\cal C}_Q[\exp (-\lambda \theta)]$ is a better lower bound. 

The second example is $\theta ={\bf 1}$, in which case
\begin{eqnarray}
{\rm Tr}\exp (-\lambda \theta)=d\exp (-\lambda);\;\;\;\;\;
\frac{1}{d}{\rm Tr}{\cal C}_Q[\exp (-\lambda \theta)]=\exp (-\lambda);\;\;\;\;\;
\frac{1}{d}\sum _{\alpha,\beta}\exp[-d\lambda Q(\alpha,\beta|\theta)]=d\exp (-\lambda).
\end{eqnarray}
Here the $\frac{1}{d}\sum _{\alpha,\beta}\exp[-d\lambda Q(\alpha,\beta|\theta)]$ is a better lower bound.

\section{Spectral formalism, POVM,  wavelets and the Choquet formalism}\label{analogy}

In this section we compare and contrast the Choquet formalism with the spectral formalism of eigenvalues and eigenvectors, the POVM formalism,  
and the formalism of frames and wavelets.
Let $\theta$ be a Hermitian operator.
\begin{itemize}
\item
{\bf Spectral formalism of eigenvalues and eigenvectors:}
\begin{itemize}
\item
It uses the complete set of the $d$ eigenvectors of $\theta$, which are orthogonal to each other.
This set is not fixed, but depends on $\theta$.
\item
$\theta=\sum e_i {\mathfrak P}_i$, where ${\mathfrak P}_i$ are the eigenprojectors and $e_i={\rm Tr}(\theta {\mathfrak P}_i)$ the eigenvalues
of $\theta$.
\item
If two operators commute, they have the same eigenprojectors ${\mathfrak P}_i$.
\end{itemize}
\item
{\bf POVM formalism:}
\begin{itemize}
\item
It uses the set $\Omega$ of $d^2$ coherent states. 
The central feature is the resolution of the identity in Eq.(\ref{1111}),
which is used in expressing various physical quantities in terms of coherent states.
\item
$\theta=\sum P(\alpha, \beta |\theta)\Pi(\alpha, \beta)$ in terms of the projectors $\Pi(\alpha, \beta)$ 
and the $P$-function $P(\alpha, \beta |\theta)$, as explained in Eq.(\ref{9}).
\end{itemize}
\item
{\bf Frames and wavelets:}
\begin{itemize}

\item
A frame is a family of states $\ket{v_i}$, such that for all (normalized) states $\ket{f}$ in the Hilbert space
\begin{eqnarray}\label{fr}
&&A\le \sum _i|\bra {v_i}f\rangle |^2\le B\nonumber\\
&&\ket {f}=\sum _i(S^{-1}\ket{v_i})\bra {v_i}f\rangle;\;\;\;\;\;S=\sum _i\ket{v_i}\bra {v_i}.
\end{eqnarray}
$A,B$ are constants called lower and upper bound.
\item
The philosophy here that if we do not know an exact resolution of the 
identity, we should try to find lower and upper bounds for it.
In this sense, the formalism uses an approximate resolution of the identity, with bounded error.
 
\end{itemize}
\item
{\bf Choquet formalism:}
\begin{itemize}
\item
It uses the set $\Omega$ of $d^2$ coherent states, but it does not use their resolution of the identity of Eq.(\ref{1111}).
The formalism introduces its own `weak resolution of the identity' of Eq.(\ref{al1}), that involves the 
non-orthogonal projectors and also the M\"obius operators that eliminate the double counting.
$Q(\alpha, \beta |\theta)=\frac{1}{d}{\rm Tr}[\theta \Pi(\alpha, \beta )]$ is
the $Q$-function of a Hermitian operator $\theta$.  
Based on the ranking in Eq.(\ref{56}), the coherent states, projectors and $Q(\alpha, \beta |\theta)$ are divided into two groups `dominant' and `inferior', which depend on $\theta$.
\item
${\cal C}_Q(\theta)=\sum Q(i|\theta) \varpi _{\theta}(i;{i+1};...;{d^2})$.
The projectors $\varpi _{\theta}(i|i+1,...,d^2)$
are discrete derivatives (differences) of the cumulative projectors $\Pi _{\theta}(i;{i+1};...;{d^2})$. The $\varpi _{\theta}(i|i+1,...,d^2)$
form an orthogonal set of $d$ projectors, and they are different from the projectors $\Pi(i)$, associated to coherent states.
The $Q(i|\theta)$ and $\varpi _{\theta}(i;{i+1};...;{d^2})$ are eigenvalues and eigenprojectors of ${\cal C}_Q(\theta)$. 
The ${\cal C}_Q(\theta)$ is a figure of merit for $\theta$, and is in general different from $\theta$.
The ${\rm Tr}[{\cal C}_Q(\theta)]$ is an upper bound for various physical quantities as shown in proposition \ref{1234} and in section \ref{partition}.

\item
If two operators are comonotonic, they have the same  $\varpi _{\theta}(i|i+1,...,d^2)$ projectors. 
Comonotonic operators have the same dominant coherent states and projectors, and their Choquet integrals commute.
Comonotonicity formalizes the vague concept of physically similar operators. 

\item

The frames and wavelets formalism, uses approximate resolutions of the identity with bounded error. 
The Choquet formalism corrects this error with the M\"obius operators, and uses the `weak resolution of the identity' of Eq.(\ref{al1}).

\end{itemize}

\end{itemize}

In this paper we used the Choquet formalism with coherent states, but as we explained 
the formalism introduces its own `weak resolution of the identity' of Eq.(\ref{al1}), and it does not use 
the resolution of the identity in Eq.(\ref{1111}).
Therefore the formalism can be used  
with total sets of states, for which we do not know explicitly a resolution of the identity
(a set of states is called total, if there is no state in the Hilbert space which is orthogonal to all states in the set).
The Choquet formalism introduces a `weak resolution of the identity', that involves the M\"obius operators
in addition to the projectors.
It is robust in the presence of noise, and it can be used as bound for various physical quantities, in the study of the 
ground state of physical systems, etc.

\section{Discussion}

The Choquet integral is used in problems with probabilities, which involve overlapping (non-independent) alternatives.
In this paper, we have used it in a quantum context with the $Q$-function of Hermitian positive semidefinite operators.
The $Q$-function is defined in terms of coherent states, which overlap with each other, and this motivates the use of this approach.
The Choquet integral uses the ranking of the values of the $Q$-function in Eq.(\ref{56}), and it is given by Eq.(\ref{010A}).

The formalism uses the M\"obius operators ${\mathfrak D} ({\alpha _1,\beta _1};{\alpha _2,\beta _2})$,
${\mathfrak D} ({\alpha _1,\beta _1};{\alpha _2,\beta _2};{\alpha _3,\beta _3})$, etc,  to quantify the overlaps between coherent states. They enter in the 
Choquet integral as described in proposition \ref{www}.
The M\"obius operators are interpreted in the context of non-additive probabilities (capacities), and they are related to commutators as in Eq.(\ref{e3}),
which shows that they are non-zero if the projectors do not commute.

A central concept in the formalism, which is novel in Physics, is comonotonicity. It is used to formalize the vague concept of physically similar operators.
Comonotonic operators are bounded as in Eq.(\ref{1qa}), with respect to the $\prec$ preorder.
This means that the values of ${\rm Tr}{\cal C}_Q(\theta)$ are bounded within a certain interval, and consequently other physical quantities 
(like ${\rm Tr} (\rho \theta)$ with any density matrix $\rho$) to which ${\rm Tr}{\cal C}_Q(\theta)$
is a bound, are also bounded.

In terms of applications,
the Choquet integral has been used to derive bounds for various physical quantities (proposition \ref{1234}, and section \ref{partition} for the partition function).
A desirable feature of the formalism, is that it is robust in the presence of noise. 
The reason is that noise affects in a uniform way all coherent states, and does not change the ranking significantly.
At the same time the formalism is sensitive enough to detect changes in the ground state of physical systems, because they affect the ranking.
Examples of this have been given in sections \ref{ex12}, \ref{DD}.

From a practical point of view, calculations are easy if they involve only the ${\rm Tr}{\cal C}_Q(\theta)$.
This simply requires the $Q$-function and its ranking in Eq.(\ref{56}) (see Eq.(\ref{rfv})).
If the full ${\cal C}_Q(\theta)$ is required, as for example in Eq.(\ref{MMM}), then the calculation of the projectors $\varpi _{\theta}(i|i+1,...,d^2)$
is needed, and this can be computationally more intensive. 

There are many figures of merit in Physics. They are used in bounds for the values of physical quantities.
They are also used to derive orders in sets of physical quantities (e.g., various entropic quantities define `more mixed' or `more entangled', etc).
In this paper we introduced the Choquet integral and the concept of comonotonicity, which are motivated by non-additive probabilities associated with overlapping alternatives, and which we have used to derive bounds to physical quantities, and study the lowest state of physical systems.

We have considered positive semidefinite operators, but the work could be extended to all Hermitian operators.
Also we have used the $Q$-function, but a similar formalism that involves the $P$-function can also be developed.
The work provides a deeper insight to the use of non-orthogonal overcomplete sets of states (like coherent states) for the 
study of physical problems.

\newpage

\begin{table}
\caption{The three dominant values of $Q(\alpha, \beta)$, the eigenvalues $e_1,e_2, e_3$, and the dominance ratio $r(\theta)$
of the operator $\theta$ in Eq.(\ref{97}). In the first rwo $r_i=0$ (there is no noise). In the other five rows
$r_i$ are uniformly distributed random numbers in the interval $(-1,1)$.
The overlaps $\tau _i=|\bra{u_i}v_i\rangle |^2$ of the eigenvectors in the noisy cases, with their counterparts in the noiseless case are also shown.}
\def\arraystretch{1.5}
\begin{tabular}{|l|l|l|l|l|l|l|l|l|l|}\hline
$Q(9|\theta)$&$Q(8|\theta)$&$Q(7|\theta)$&$e_1$&$e_2$&$e_3$&$r(\theta)$&$\tau _1$&$\tau _2$&$\tau _3$\\\hline
$Q(1,2)=3.023$&$Q(1,1)=3.023$&$Q(0,2)=2.095$&$0.942$&$5.488$&$12.569$&$0.428$&$1$&$1$&$1$\\\hline
$Q(1,2)=3.447$&$Q(1,1)=2.926$&$Q(0,2)=2.171$&$0.604$&$4.993$&$13.235$&$0.453$&$0.973$&$0.979$&$0.992$\\
$Q(1,2)=3.173$&$Q(1,1)=2.897$&$Q(0,0)=2.398$&$1.337$&$6.743$&$12.230$&$0.416$&$0.987$&$0.990$&$0.996$\\
$Q(1,1)=2.911$&$Q(1,2)=2.506$&$Q(0,1)=1.865$&$0.809$&$4.245$&$11.380$&$0.443$&$0.990$&$0.985$&$0.980$\\
$Q(1,2)=3.157$&$Q(1,1)=2.962$&$Q(0,2)=2.278$&$0.747$&$4.454$&$13.111$&$0.458$&$0.997$&$0.988$&$0.988$\\
$Q(1,2)=3.316$&$Q(1,1)=3.180$&$Q(0,1)=2.436$&$0.836$&$5.774$&$13.671$&$0.413$&$0.967$&$0.950$&$0.980$\\\hline
\end{tabular} \label{t1}
\end{table}

\begin{table}
\caption{Comonotonicity intervals, the corresponding three dominant values of $Q(\alpha, \beta)$, and the 
dominance ratio $r[\theta (\lambda)]$, 
for the operator $\theta (\lambda)$ in Eq.(\ref{1000}).}
\def\arraystretch{2}
\begin{tabular}{|l||l|l|l|l|}\hline
{\rm intervals of\;\;}$\lambda$ &$Q[9|\theta (\lambda)]$&$Q[8|\theta (\lambda)]$&$Q[7|\theta(\lambda)]$&$r[\theta (\lambda)]$\\\hline
$I_1=(0,0.06)$&$Q(1,0)$&$Q(0,0)$&$Q(2,0)$&$\frac{13.5+7\lambda}{99+63 \lambda}$\\\hline
$I_2=(0.06, 0.44)$&$Q(0,0)$&$Q(1,0)$&$Q(2,0)$&$\frac{13.5+7\lambda}{99+63 \lambda}$\\\hline
$I_3=(0.44,0.56)$&$Q(0,0)$&$Q(1,0)$&$Q(0,2)$&$\frac{12.62+9\lambda}{99+63 \lambda}$\\\hline
$I_4=(0.56, 0.6)$&$Q(0,0)$&$Q(0,2)$&$Q(1,0)$&$\frac{12.62+9\lambda}{99+63 \lambda}$\\\hline
$I_5=(0.6,0.7)$&$Q(0,0)$&$Q(0,2)$&$Q(0,1)$&$\frac{11+11.7\lambda}{99+63 \lambda}$\\\hline
\end{tabular} \label{t2}
\end{table}
\begin{table}
\caption{The three dominant values of $Q[\alpha, \beta|{\mathfrak P}(\lambda)]$, and the dominance ratio $r[{\mathfrak P}(\lambda)]$ as a function of $\lambda$, 
for the ${\mathfrak P}(\lambda)=\ket{g(\lambda)}\bra{g(\lambda)}$
where $\ket{g(\lambda)}$ is the ground state of the system described with the Hamiltonian ${\mathfrak H}(\lambda)$ in Eq.(\ref{ham}).
The Wehrl entropy $E[{\mathfrak P}(\lambda)]$ (in nats) and the $|\bra {g(0)}g(\lambda)\rangle|^2$ are also shown.
Horizontal lines indicate that we cross from one equivalence class to another}
\def\arraystretch{1.5}
\begin{tabular}{|l||l|l|l|l|l|l|}
  \hline$\lambda$& $Q[9|{\mathfrak P}(\lambda)]$ & $Q[8  |{\mathfrak P}(\lambda)]$ &$Q[7 |{\mathfrak P}(\lambda)]$ 
& $r[{\mathfrak P}(\lambda)]$&$E[{\mathfrak P}(\lambda)]$ &$|\bra {g(0)}g(\lambda)\rangle|^2$\\\hline
$0.0$& $Q[1,2|{\mathfrak P}(\lambda)]=0.214$ & $Q[1,0|{\mathfrak P}(\lambda)]=0.214$ &$Q[1,1|{\mathfrak P}(\lambda)]=0.214$ &$0.642 $&$1.929$&$1$\\
$0.1$& $Q[1,2|{\mathfrak P}(\lambda)]=0.228$ & $Q[1,0|{\mathfrak P}(\lambda)]=0.209$ &$Q[1,1|{\mathfrak P}(\lambda)]=0.191$ &$0.628 $&$1.948$&$0.971$\\
$0.2$& $Q[1,2|{\mathfrak P}(\lambda)]=0.245$ & $Q[1,0|{\mathfrak P}(\lambda)]=0.202$ &$Q[1,1|{\mathfrak P}(\lambda)]=0.166$ &$0.613 $&$1.972$&$0.929$\\
$0.3$& $Q[1,2|{\mathfrak P}(\lambda)]=0.268$ & $Q[1,0|{\mathfrak P}(\lambda)]=0.192$ &$Q[1,1|{\mathfrak P}(\lambda)]=0.138$ &$0.598 $&$1.980$&$0.885$\\\hline
$0.4$& $Q[1,2|{\mathfrak P}(\lambda)]=0.297$ & $Q[1,0|{\mathfrak P}(\lambda)]=0.176$ &$Q[2,2|{\mathfrak P}(\lambda)]=0.120$ &$0.593  $&$1.959$&$0.812$\\\hline
$0.5$& $Q[1,2|{\mathfrak P}(\lambda)]=0.313$ & $Q[0,2|{\mathfrak P}(\lambda)]=0.168$ &$Q[2,2|{\mathfrak P}(\lambda)]=0.157$ &$0.638 $&$1.889$&$0.666$\\
$0.6$& $Q[1,2|{\mathfrak P}(\lambda)]=0.298$ & $Q[0,2|{\mathfrak P}(\lambda)]=0.216$ &$Q[2,2|{\mathfrak P}(\lambda)]=0.187$ &$0.701 $&$1.816$&$0.477$\\
$0.7$& $Q[1,2|{\mathfrak P}(\lambda)]=0.268$ & $Q[0,2|{\mathfrak P}(\lambda)]=0.242$ &$Q[2,2|{\mathfrak P}(\lambda)]=0.200$ &$0.710 $&$1.825$&$0.330$\\\hline
$0.8$& $Q[0,2|{\mathfrak P}(\lambda)]=0.252$ & $Q[1,2|{\mathfrak P}(\lambda)]=0.241$ &$Q[2,2|{\mathfrak P}(\lambda)]=0.202$ &$0.695 $&$1.845$&$0.240$\\
$0.9$& $Q[0,2|{\mathfrak P}(\lambda)]=0.256$ & $Q[1,2|{\mathfrak P}(\lambda)]=0.221$ &$Q[2,2|{\mathfrak P}(\lambda)]=0.200$ &$0.677 $&$1.862$&$0.185$\\
$1.0$& $Q[0,2|{\mathfrak P}(\lambda)]=0.257$ & $Q[1,2|{\mathfrak P}(\lambda)]=0.207$ &$Q[2,2|{\mathfrak P}(\lambda)]=0.198$ &$0.662 $&$1.873$&$0.149$\\\hline
\end{tabular}\label{t3}
\end{table}

\begin{table}
\caption{The three eigenvalues $e_1,e_2,e_3$ of the Hamiltonian in Eq.(\ref{hamilt}).
The three dominant values of $Q[\alpha, \beta|{\mathfrak P}(\lambda)]$ for the ${\mathfrak P}(\lambda)=\ket{g(\lambda)}\bra{g(\lambda)}$
where $\ket{g(\lambda)}$ is the ground state of the system, are also shown.
In the case $\lambda=0$ the two lowest eigenvalues are equal to each other, and the ${\mathfrak P}(0)$ is the projector to the corresponding two-dimensional eigenspace.
In this case we present the  $Q[\alpha, \beta|\frac{1}{2}{\mathfrak P}(0)]$}
\def\arraystretch{1.5}
\begin{tabular}{|l||l|l|l|l|l|l|l|}
  \hline$\lambda$& $e_1$&$e_2$&$e_3$&$Q[9|{\mathfrak P}(\lambda)]$&$Q[8|{\mathfrak P}(\lambda)]$&$Q[7|{\mathfrak P}(\lambda)]$&$E[{\mathfrak P}(\lambda)]$\\\hline
$-0.01$& $0.494$&$0.509$&$6.495$&$Q[0,1|{\mathfrak P}(-0.01)]=0.194$&$Q[0,2|{\mathfrak P}(-0.01)]=0.194$&$Q[1,1|{\mathfrak P}(0.01)]=0.163$&$1.960$\\\hline
$0.00$& $0.500$&$0.500$&$6.500$&$Q[1,1|\frac{1}{2}{\mathfrak P}(0)]=0.160$&$Q[1,2|\frac{1}{2}{\mathfrak P}(0)]=0.160$&$Q[2,1|\frac{1}{2}{\mathfrak P}(0)]=0.155$&$2.088 $\\\hline
$0.01$& $0.490$&$0.505$&$6.505$&$Q[2,1|{\mathfrak P}(0.01)]=0.218$&$Q[2,2|{\mathfrak P}(0.01)]=0.218$&$Q[1,1|{\mathfrak P}(0.01)]=0.155$&$1.890 $\\\hline
$-0.01i$& $0.492$&$0.507$&$6.500$&$Q[2,2|{\mathfrak P}(-0.01i)]=0.292$&$Q[1,2|{\mathfrak P}(-0.01i)]=0.291$&$Q[0,2|{\mathfrak P}(-0.01i)]=0.255$&$1.627 $\\\hline
$0.01i$& $0.493$&$0.507$&$6.500$&$Q[2,1|{\mathfrak P}(0.01i)]=0.292$&$Q[1,1|{\mathfrak P}(0.01i)]=0.291$&$Q[0,1|{\mathfrak P}(0.01i)]=0.255$&$1.627 $\\\hline
\end{tabular}\label{t4} 
\end{table}

\end{document}